\documentclass[11pt, a4paper, twoside]{article}
\usepackage[utf8]{inputenc}
\usepackage[english]{babel}
\usepackage{amsmath}
\usepackage{amsthm}
\usepackage{amsfonts}
\usepackage{amssymb}
\usepackage{amsthm}
\usepackage{pgfplots}
\usepackage{caption}
\usepackage{subcaption}
\usepackage{bbm}
\usepackage[shortlabels]{enumitem}
\usepackage{float}
\usepackage{stackengine}
\usepackage[margin=1in]{geometry}
\usepackage{accents}
\usepackage{tocloft}
\tocloftpagestyle{empty}
\usepackage{appendix}
\usepackage{xcolor}
\definecolor{darkred}{rgb}{0.5,0,0}
\definecolor{lightblue}{rgb}{0,0.4,0.8}
\definecolor{darkgreen}{rgb}{0,0.5,0}
\usepackage{listings} 
\usepackage{xfrac}
\usepackage{comment}
\usepackage[export]{adjustbox}

\newcounter{sideremark}

\definecolor{dkgreen}{rgb}{0,0.6,0}
\definecolor{gray}{rgb}{0.5,0.5,0.5}
\definecolor{mauve}{rgb}{0.58,0,0.82}

\usepackage[citecolor=blue]{hyperref}
\hypersetup{
    colorlinks=true,
    linkcolor=blue,
    urlcolor=blue,
}
\usepackage{graphicx}
\hypersetup{colorlinks=true, linktoc=all, linkcolor=blue}
\usetikzlibrary{arrows.meta}

\pgfplotsset{compat=newest}
\DeclareMathOperator{\Unif}{Unif}
\DeclareMathOperator{\E}{\mathbb{E}}
\DeclareMathOperator{\Prob}{\mathbb{P}}

\newcommand\smallO{
	\mathchoice
	{{\scriptstyle\mathcal{O}}}
	{{\scriptstyle\mathcal{O}}}
	{{\scriptscriptstyle\mathcal{O}}}
	{\scalebox{.7}{$\scriptscriptstyle\mathcal{O}$}}
}
\newcommand*{\defeq}{\mathrel{\vcenter{\baselineskip0.5ex \lineskiplimit0pt
			\hbox{\scriptsize.}\hbox{\scriptsize.}}}%
	=}
\newcommand{\aval}{0.789}
\newcommand{\bval}{1.24}
\newcommand{\pval}{0.421}
\newcommand{\rat}{0.72349}
\newcommand{\CR}{0.7235}

\newcommand\ndots{\makebox[1em][c]{.\hfil.\hfil.}}

\usepackage[capitalize]{cleveref}
\crefname{equation}{}{}
\newtheorem{theorem}{Theorem}[section]
\newtheorem{proposition}{Proposition}[section]
\newtheorem{lemma}{Lemma}[section]

\newtheorem{remark}{Remark}[section]
\newtheorem{definition}{Definition}[section]
\numberwithin{equation}{section}
\title{Prophet Inequalities:\\Separating Random Order from Order Selection}
\author{Giordano Giambartolomei\setcounter{footnote}{1}\thanks{King's College London.}, Frederik Mallmann-Trenn\footnotemark[2] and Raimundo Saona\setcounter{footnote}{7}\thanks{Institute of Science and Technology Austria.}
}
\date{}
\begin{document}
    \maketitle
    \thispagestyle{empty}
    \begin{abstract}
    	Prophet inequalities are a central object of study in optimal stopping theory. A gambler is sent values in an online fashion, sampled from an instance of independent distributions, in an adversarial, random or selected order, depending on the model. When observing each value, the gambler either accepts it as a reward or irrevocably rejects it and proceeds to observe the next value. The goal of the gambler, who cannot see the future, is maximising the expected value of the reward while competing against the expectation of a prophet (the offline maximum). In other words, one seeks to maximise the gambler-to-prophet ratio of the expectations. 
    	
    	The model, in which the gambler selects the arrival order first, and then observes the values, is known as Order Selection. In this model a ratio of $0.7251$ is attainable for any instance. Recently, this has been improved up to $0.7258$ by Bubna and Chiplunkar (2023). If the gambler chooses the arrival order (uniformly) at random, we obtain the Random Order model. The worst case ratio over all possible instances has been extensively studied for at least $40$ years. In a computer-assisted proof, Bubna and Chiplunkar (2023) also showed that this ratio is at most $0.7254$ for the Random Order model, thus establishing for the first time that carefully choosing the order, instead of simply taking it at random, benefits the gambler. We give an alternative, non-simulation-assisted proof of this fact, by showing mathematically that in the Random Order model, no algorithm can achieve a ratio larger than $\CR$. This sets a new state-of-the-art hardness for this model, and establishes more formally that there is a real benefit in choosing the order.
    \end{abstract}
    \clearpage
    \pagenumbering{arabic}
    \setcounter{footnote}{0}
	\section{Introduction}
	Prophet inequalities are a central object of study in optimal stopping theory. A gambler is sent nonnegative values in an online fashion, sampled from an instance of uniformly bounded independent random variables $\{V_i\}$ with known distributions $\{\mathcal{D}_i\}$, in adversarial, random or selected order, depending on the particular model. When observing each value, the gambler either accepts it as a reward, or irrevocably rejects it and proceeds with observing the next value. The goal of the gambler, who cannot see the future, is to maximise the expected value of the reward while competing against the expectation of a prophet (out of metaphor, the offline maximum or supremum, depending on whether the instance is finite or not). In other words, one seeks to maximise the gambler-to-prophet ratio of the expectations. The gambler represents any online decision maker, such as an algorithm or stopping rule. Probabilistically, we will refer to it as a \emph{stopping time} $\tau$: being online implies not being able to see the future, thus the gambler will always stop at a time $\tau$ such that the event $\{\tau=i\}$ depends, informally speaking, only on the first $i$ values observed.
	
    \subsection{Models of prophet inequalities}
	
    To date, several models and extensions of prophet inequalities have appeared in the literature. In this section we first introduce the basic terminology, the classical model, and then describe three relaxations of it, briefly reviewing the state-of-the-art concerning their hardness and the competitiveness of known algorithms. More details are provided in \Cref{related}.
    
    \subsubsection{Terminology}
    Due to the online nature of prophet inequalities, some terminology from \emph{competitive analysis} is usually borrowed. Since conventions can vary, we summarise it briefly. We refer to a worst-case gambler-to-prophet ratio of the algorithm (that is a ratio known to be achievable for any given instance by the algorithm) as \emph{competitive ratio}. Let $c\in[0,1]$ be a real constant. An algorithm is said to be $c$-competitive if it has a competitive ratio of $c$, meaning that it can attain at least a gambler-to-prophet ratio of $c$ for any instance, but it does not necessarily mean that $c$ is the highest possible ratio. An upper bound on any algorithm's highest possible competitive ratio for a given instance will be called \emph{hardness} of the instance. Saying that a prophet inequality model is $c$-hard or has a hardness of $c$, means that there is a $c$-hard instance for that problem, but it does not necessarily mean that $c$ is the lowest hardness possible amongst all instances for that model. A hardness for a model is said to be \emph{tight} (or \emph{optimal}) when it is matched by the competitive ratio of an algorithm solving it. Similarly, the competitive ratio of an algorithm solving a model is tight (or optimal) when it is matched by a hardness known for that model. Often, when determining a hardness or a competitive ratio, numerical computations are involved. Thus tightness can be used in a broad sense. For two models $A$ and $B$, we say that $A$ \emph{beats} $B$ if the hardness of $B$ is strictly less than the competitive ratio of an algorithm solving $A$. When $A$ beats $B$ or $B$ beats $A$, we say that $A$ and $B$ are \emph{separated}.
	
    \subsubsection{The classical Prophet Inequality}
	The very first model of prophet inequality, typically referred to as (adversarial) Prophet Inequality (PI), is due to Krengel and Sucheston \cite{KrenSuch77,KrenSuch78}. The given instance is composed of countably many independent random variables $\{V_i\}$ with known distributions $\{\mathcal{D}_i\}$ in a fixed given order, usually referred to as \emph{adversarial} order. The general dynamics aforementioned is followed: the gambler observes in an online fashion the sequence of sampled values, and makes irrevocable decisions to accept a value and stop, or continue observing, with the goal of maximising the expected value of the reward, while competing against the expectation of a prophet, which represents the offline supremum. More precisely, the goal is to maximise the gambler-to-prophet ratio, that is $\E V_\tau$ over $\E\sup_iV_i$, where $\tau$ is the stopping time associated to the gambler's stopping rule. In~\cite{KrenSuch78}, it was shown that the $\sfrac{1}{2}$-hardness of PI (shown by Garlin \cite{KrenSuch77}) is tight. Later it was shown that a competitive ratio of $\sfrac{1}{2}$ can be attained even by a \emph{single threshold} algorithm \cite{Sam84}. The classical PI has subsequently been relaxed giving more power to the gambler, which leads to larger competitive ratios.
	
	\subsubsection{IID Prophet Inequality, Order Selection and Random Order models}
    \label{models}
 	In this section we introduce our working assumptions and some basic notation, and then review the three main models of prophet inequalities related to our work. We will restrict ourselves to \emph{finite} instances, that is, we will always work with a size $n\in\mathbb{N}$ fixed and with non-negative random variables $V_1,\ldots,V_n$ having distributions $\mathcal{D}_1,\ldots,\mathcal{D}_n$. Therefore, it is natural that, if the gambler gets to the last stage, any value observed will have to be accepted. Thus the stopping time $\tau$ of the gambler will belong to the class of stopping times valued in $[n]\defeq\{k\in\mathbb{N}:\:k\leq n\}=\{1,2,\ldots,n\}$. We denote this class as $C^n$. Also, we denote the set of permutations of $[n]$ by $S_n$, and for any $\pi\in S_n$, we will adopt the \emph{inline} notation $\pi=(\pi_1,\ldots,\pi_n)$.
    
 	The following three variants of PI are related to our work. For readability, the following exact values will be denoted by their fourth decimal approximation, followed by dots, and viceversa, whenever using these decimal values, follwed by dots, we mean the corresponding exact constants: $1 - \sfrac{1}{e} = 0.6321\ndots$; $\sqrt{3}-1= 0.7321\ndots$; $\sfrac{1}{\beta}= 0.7451\ndots$, where $\beta\in\mathbb{R}$ is the unique solution to $\int_0^1 \frac{dx}{x(1-\log x)+\beta-1} = 1$.
 	
    \paragraph{IID PI.} A specialised case of PI, where the random variables are assumed independent and identically distributed (iid). The hardness of this problem has been shown to be $0.7451\ndots$ in~\cite{Kertz86}. That this is tight follows from a $0.7451\ndots$-competitive \emph{quantile strategy} devised in~\cite{Correa17}. Remarkably, the quantiles generating the thresholds do not depend on the distribution.
    
	\paragraph{Order Selection (OS).} Also known as \emph{free order}, OS is a variant where the gambler is allowed to \emph{select} the arrival order $\pi \in S_{n}$ first, and then observe the values sampled from $V_{\pi_1}, \ldots, V_{\pi_n}$, thus seeking to maximise the ratio of $\E V_{\pi_\tau}$ over $\E\max_{i\in[n]}V_i$, where the randomness is with respect to $\tau$ only. This version, introduced in \cite{Hill83}, is also central in the study of revenue maximisation in Posted Price Mechanisms~\cite{Einav18}. After various improvements (see \cref{related}) on the competitive ratio of $0.6321\ndots$ shown in~\cite{Chawla10}, the state-of-the-art for OS is the $0.7258$ competitive ratio obtained in \cite{BubChip23} by refining the $0.7251$-competitive algorithm exploiting \emph{continuous time arrival design} introduced in~\cite{PeTa22}. The $0.7451\ndots$-hardness of OS follows directly from the hardness of IID PI.
	
	\paragraph{Random Order (RO).} Also known as \emph{prophet secretary}, RO is a variant where random variables are shown in a uniform random order $\pi$ to the gambler, who observes the values sampled from $V_{\pi_1}, \ldots, V_{\pi_n}$ and seeks to maximise the ratio of $\E V_{\pi_\tau}$ over $\E\max_{i\in[n]}V_i$, where the randomness is with respect to both $\pi$ and $\tau$.
        One can equivalently see this model as arising from the OS setting, by saying that the gambler chooses the order $\pi$ uniformly at random in $S_n$. After various improvements (see \cref{related}) on the initial 
        competitive ratio of $0.6321\ndots$ achieved in~\cite{Esf17}, where the model was introduced, the state-of-the-art has been achieved in \cite{Harb24}, through a $0.6724$-competitive continuous version of the $0.6697$-competitive multi-threshold algorithms known as \emph{blind strategies} introduced in \cite{CorrSaZil21}. The $0.7321\ndots$-hardness of RO, proved analytically in \cite{CorrSaZil21}, has been recently improved in \cite{BubChip23}, where a new $0.7254$-hardness for the model is established in a computer-assisted proof.

	\subsection{Applications to pricing}
	Prophet inequalities are closely related to Posted Price Mechanisms (PPMs), which are an attractive alternative to implementing auctions and are often used in online sales~\cite{Einav18}. The way these mechanisms work is as follows. Suppose a seller has an item to sell. Customers arrive one at a time and the seller proposes to each customer a take-it-or-leave-it price. The first customer accepting the offer pays the price and takes the item \cite{Correa19}. 
    
    If a seller faces buyers with private information about their willingness to pay and there are no further transaction costs, an auction is optimal \cite{Einav18,HarrTown81,Myer81,RilSam81}. However, auctions can have high transaction costs. They take time and require communication with multiple buyers. There are several circumstances under which price posting may be preferable, being much simpler, yet efficient enough \cite{Einav18,Correa19}. 
    
    Since PPMs are suboptimal, it is important to know the ratio between PPMs and the optimal auction (Myerson's auction) \cite{Correa19}. This ratio can be studied from the point of view of prophet inequalities: designing PPMs can be translated into designing algorithms for prophet inequality models \cite{Haji07,Chawla10} and viceversa \cite{Correa19}. In particular, OS readily connects to ordering potential buyers in the PPM setting, while RO can be seen as taking potential buyers uniformly at random. Therefore, it is important to know if ordering the potential buyers can lead to improved performance, compared to a uniform random order. Since the introduction of OS $40$ years ago, there has been significant effort spent on finding how far the benefits of ordering go: for PPMs showing that ordering yields benefits is equivalent to separating RO from OS.
	
	\subsection{Our contributions}
	Our main contribution consists of improving on the hardness of RO, to the extent of separating it from OS, through rigorous mathematical analysis. The new hardness of $\CR$ follows from the asymptotic analysis of the online optimal algorithm applied to a simple instance, which we introduce next.
	
	\subsubsection{The hard instance}\label{hardinstance}
	The hard instance consists of $n$ iid random variables and a constant.
	\begin{definition}
		Let $b=\bval$, $p=\pval$ and fix $n\in\mathbb{N}$ large enough, so as to have a well defined random variable
		\[V\sim
		\begin{cases}
			n,&\text{w.p. }\:\sfrac{1}{n^2}\\
			b,&\text{w.p. }\:\sfrac{p}{n}\\
			0,&\text{w.p. }\: 1 - \sfrac{p}{n} - \sfrac{1}{n^2}.
		\end{cases}\]
		We define the instance $\{V_1,V_2,\ldots,V_{n+1}\}$, which consists of $n$ iid random variables distributed as $V$ and a degenerate one, equal to the constant $a=\aval$, that is
		\begin{equation}\tag*{H}\label{instance}
			V_i\sim\begin{cases}V,&1\leq i\leq n\\ a,&i=n+1.\end{cases}
		\end{equation}
	\end{definition}
	
	\subsubsection{Main results}
	To establish a hardness for RO, it is enough to show a uniform upper bound on the competitive ratio of all stopping rules, that is an upper bound, holding for some given instance and uniformly for all stopping rules $\tau$, on the ratio of $\E V_{\pi_\tau}$ over $\E\max_{i}V_i$. This is done by upper-bounding the expected reward $\E V_{\pi_T}$ of an \emph{optimal} stopping rule $T$, which is a stopping rule (existing by \emph{backward induction}) maximising the expected reward.
	\begin{theorem}\label{mainth}
		On Instance \ref{instance}, the optimal stopping rule $T\in C^{n+1}$ is such that
		\[\limsup_{n\longrightarrow\infty}\frac{\E V_{\pi_T}}{\E\max_{i\in[n+1]}V_i}< \CR.\]
        As a consequence, RO is $\CR$-hard. Therefore OS is separated from RO, and the former beats the latter. 
	\end{theorem} 
	
	\subsubsection{Our techniques}
	To derive mathematically a new state-of-the-art hardness of RO, we rely on an innovative asymptotic analysis of the optimal algorithm's acceptance thresholds, computed via backward induction, in a random arrival order setting, so as to obtain asymptotic bounds on the competitive ratio of the optimal algorithm that have, in principle, the potential of applying to tight instances for RO. A key role in this analysis is played by the study of acceptance times for the values $a$ and $b$, depending on the possible histories of the process. The simplicity of Instance \ref{instance} ensures, by design, that there will be only one acceptance time for $a$, and two for $b$, depending on whether the history of the process has already rejected $a$ or not (see \Cref{acceptance} for details). 
    
    Future applications of these analytic techniques are not limited to the hard instance studied, nor to RO. In fact, Instance \ref{instance} is found by analysing a much larger class of instances (see \Cref{application} for details), which can be characterised as a collection of a positive constant and $n$ iid random variables (we will informally refer to $n$ as the size of the instance), supported on a set of three nonnegative values:
    \begin{itemize}[noitemsep] 
    \item a size-dependent value always accepted by the optimal algorithm, which gets larger but also disproportionately more unlikely as $n$ grows; 
    \item a null value, which is never accepted until the final step, and which gets increasingly likely as $n$ grows;
    \item a positive value, which gets more unlikely (but to a lesser extent than the size-dependent value) as $n$ grows, and whose acceptance depends on the past values and the stage the algorithm is at; we thus refer to this value as \emph{nontrivial}. 
    \end{itemize}
    
    This class of instances, which we analyse as the size grows, is thus fairly general and broadly speaking it contains the $0.7321\ndots$-hard instance studied in~\cite[\S7]{CorrSaZil21}, which is composed of $n$ iid random variables taking value $n$ with probability $\sfrac{1}{n^2}$ and zero otherwise, and an additional positive constant $a$, which is a nontrivial value for the optimal algorithm. Our class is obtained by adding to the iid distributions an increasingly unlikely point mass at another nontrivial value.
	
	\subsection{Previous bounds on Random Order}\label{upperbounds}
	To fully understand the significance of previous bounds on RO, it is best to distinguish between two versions of the model implicitly appearing in the literature (thus the terminology adopted here is nonstandard). 
	\begin{itemize} [noitemsep]
		\item \textbf{Undisclosed RO}. Also known as \emph{anonymous} RO, this is the original version, enforcing \emph{undisclosed} uniform random order, as per the definition given in \cite{Esf17}: at time $i$, only the realised value of $V_{\pi_i}$ is sent to the gambler, who in general will never know the label $\pi_i$ of the distribution it has been sampled from (in particular cases the \emph{a priori} knowledge of the distributions and the history of the probed values allow the gambler to deduce it). The $0.6321\ndots$-competitive multi-threshold algorithm designed in \cite{Esf17} is still the state-of-the-art for this model, and it is not known whether it is $0.6321\ndots$-hard or not.
		\item \textbf{Disclosed RO}. Also known as \emph{personalised} RO, this version, whose introduction is motivated by PPMs, where prices can be personalised, allows for a \emph{disclosed} random order: at time $i$, the realised value $V_{\pi_i}$ is sent to the gambler together with the label $\pi_i$, revealing what distribution the sample observed comes from. Thus the gambler acquires information about the order, as it observes more values. All the algorithms surpassing $0.6321\ndots$-competitiveness for RO implicitly refer to this model, which is therefore the most commonly referred to in the more recent literature. It is also a more natural point of view when seeing RO within the OS framework (where every selected label $\pi_i$ is clearly known in advance).
	\end{itemize}
	
	For the gambler, undisclosed RO is no easier than disclosed RO, since in the undisclosed model less information is provided.\footnote{A similar terminology has been introduced in \cite{BubChip23}, although referring to different models: \emph{order-aware} and \emph{order-unaware} RO. To avoid confusion, we stress that the \emph{order-aware} model does not refer to disclosed RO; rather, it is a model of RO, where the algorithm is partially offline, due to knowing both the past and the future labels $\pi_i$. Equivalently, the order is known in advance not only by the prophet, but by the gambler as well. Our analysis is not concerned with this recently introduced variant of RO. However, our analytic techniques could easily be extended to study the optimal algorithm for this variant on Instance \ref{instance}. 
 In~\cite[\S1.1]{BubChip23}, the authors refer to the present work as using \emph{the same idea} as theirs. Indeed, both works generalise, in different and independent ways, the hard instance presented in~\cite{CorrSaZil21}. Apart from this common origin of the hard instance, this work shares very little with theirs. Techniques and methodologies are very different. We conduct a theoretical and analytic study of RO directly, while they leverage a simulation-assisted analysis of the expected return of order-aware RO.
	} However, it is not known if there is a real benefit in disclosing the order, that is if the two versions of RO are separated. Although in undisclosed RO the algorithm does not know what distribution the value shown comes from \emph{in general}, for specific instances it is possible to deduce this information, fully or partially, from the \emph{a priori} knowledge of the distributions given and the history of the probed values. In particular, for some instances, disclosed and undisclosed RO can be equivalent. This is precisely the case for Instance \ref{instance}, as the support of all distributions is given as \emph{a priori} knowledge to the algorithm: values $0$, $b$, $n$ are coming from $V_1,\ldots,V_n\sim V$, so probing them at step $k$ means that $\pi_k(\omega)\in[n]$, while probing $a$ at step $k$ means that $\pi_k(\omega)=n+1$. Thus the order oblivious stopping rule has, for Instance \ref{instance}, a way of extracting partial information about $\pi$ from the history of the probed values. Although partial, distributionally the inferred information is equivalent to the one that the disclosed RO model would allow: indeed besides $V_{n+1}$, all other distributions are iid and knowing also in which order they arrived in the past offers no advantage in expectation. Therefore, the distinction between the two RO models does not affect the analysis in the present work. For simplicity, its set-up being more straightforward from a measure-theoretic point of view, we will adopt the simpler framework of undisclosed RO, and the hardness achieved will nonetheless hold also for disclosed RO.
	
	For the gambler, RO is no easier than OS: a uniform random permutation can always be selected. Therefore, the competitive ratio achievable by an algorithm with OS is at least the one achievable with RO. However, this does not ensure that the gambler can benefit from selecting the order, or equivalently, that RO and OS are separated.
	
	Since its introduction in~\cite{Esf17}, the only non trivial bound for RO was the $0.7451\ndots$-hardness of the IID PI model (which can be seen as a special case) proved in~\cite{Kertz86} almost $30$ years earlier. Unlike the design of algorithms for RO, which made a relatively quick progress (see \Cref{related}), it was not until $6$ years later that the $0.7321\ndots$-hard instance previously described was found and mathematically analysed in~\cite[\S7]{CorrSaZil21}. On the other hand, in~\cite{PeTa22} it was proved that the competitive ratio for OS is at least $0.7251$. Refining the techniques introduced in~\cite{PeTa22}, this ratio was recently improved in~\cite{BubChip23} to a lower bound of $0.7258$. These mathematical arguments alone fall short of proving that RO and OS are separated.
	
	In \cite[\S6]{BubChip23} a brute force simulation is provided, which improves on the hardness of RO enough to show that RO is indeed separated from OS. The simulation consists of a dynamic program, run on a $26$-parameters instance, yielding a $0.7254$-hardness for RO, which is less than $0.7258$, the state-of-the-art competitive ratio of OS. 
	From a mathematical standpoint, these numerics could, in principle, be extrapolated into a rigorously proven upper bound through formal error analysis, but this has not been carried out. Furthermore, the computed upper bound of $0.7254$ for RO differs from the improved lower bound for OS only by $0.0004$. A study of the propagation of errors in the numerics involved would be beneficial. On the other hand, due to the shear amount of parameters, the numerically found instance is hardly formally tractable, and the resulting computational implementation is quite involved. Since the last step of this argument relies on the heuristic confidence, which we share, that the decimals provided are unaffected, we deem it to be still valuable to find a mathematical argument establishing the result. 
	
	In this paper we provide such formal argument, as we focus on shedding light on the inner workings of RO, rather than on OS. We pinpoint a much simpler instance, which yields a significantly improved state-of-the-art hardness for RO of $0.7235$. At the same time this provides a straightforward and independent proof, since it comes without the need for further improving on the competitive ratio of $0.7251$ for OS. Numerically, we rely on a root finding routine \textit{only} to compute explicitly the value of the theoretical hardness obtained, which is formally expressed in terms of the zero of a transcendental function. In addition to providing a stronger bound, our technique has, in principle, the potential of yielding tight instances for RO, whereas this is provably not possible for the simulation of \cite{BubChip23}[p. 305].
    
	\subsection{Related work}\label{related}
	In this section, we briefly review some of the vast literature on prophet inequalities, with an emphasis on IID PI, OS and RO. There is in fact a variety of extensions of prophet inequalities to combinatorial structures such as matroids, matchings and combinatorial auctions: the reader interested in these aspects is better served by surveys such as \cite{Correa18}, which contains many recent developments, as well as \cite{Lucier17} for an economic point of view (especially regarding PPMs); \cite{HillKertz92} is also an interesting survey, containing classical results concerning infinite instances. In \Cref{models,upperbounds} we already explored the literature on hardness results, so we will use this section to review some of the developments on lower bounds, that is the progress made on the competitive ratio of the various algorithms designed for these models. An optimal algorithm can be found by backward induction. However, it is of little practical use outside hardness arguments, thus the research has focused on designing similar (that is, threshold-based) but simpler algorithms.
    \paragraph{IID PI.} The first to surpass the $\sfrac{1}{2}$ barrier of PI was a $0.6321\ndots$-competitive algorithm based on complicated recursive functions designed in \cite{HillKertz82}; $35$ years later an approximate single threshold $0.7380$-competitive algorithm was found in \cite{Abolhass17}, after which the tight optimal competitive ratio of $0.7451\ndots$ was finally attained as reviewed in \Cref{models}.
    \paragraph{OS.} The first to surpass the $\sfrac{1}{2}$ barrier of PI were certain PPMs proved to be $0.6321\ndots$-competitive in \cite{Chawla10} (recall that PPMs and prophet inequalities are equivalent). After $8$ years, further improvement came from a $0.6346$-competitive algorithm actually designed for disclosed RO (recall that RO is no easier than OS) in \cite{AzarChipKap18}. Shortly after, a $0.6541$-competitive PPM appeared in \cite{Bey21}. Further improvement came successively from another algorithm designed for disclosed RO: the already mentioned $0.6697$-competitive algorithm introduced in \cite{CorrSaZil21}. Thus we reached the state-of-the-art of the $0.7251$-competitive algorithm designed in \cite{PeTa22}, which has been recently polished, so as to be $0.7258$-competitive in \cite{BubChip23}. There are also special cases for which algorithms can achieve an even better competitive ratio: in \cite{Abolhass17} it was shown that when each distribution in the instance occurs $m$ times, a $0.7380$-competitive algorithm can be designed for $m$ large enough. No progress has been made on the hardness of OS.\footnote{This is due to the fact that it requires finding the optimal ordering, and in \cite{AgraSetZha20} this has been shown to be NP-hard, even under a special case where all the distributions have $3$-point support, and where the highest and lowest points of the support are the same for all the distributions.}
    \paragraph{RO.} We already mentioned that, as of yet, no algorithm, which has surpassed the initial $0.6321\ndots$-competitiveness established in \cite{Esf17}, has been designed for undisclosed RO. In fact the algorithms in \cite{AzarChipKap18,CorrSaZil21}, mentioned in the previous point, are designed and directly analysed for continuous distributions. This analysis is compatible with undisclosed RO; however, if point masses are allowed, the analysis is extended through a standard argument exploiting stochastic tie breaking, which requires revealing the distributional identity of the random variable from which the samples probed come from. Hence stochastic tie breaking can only be performed with disclosed RO. In this sense, not only personalisation has been used to go beyond $0.6321\ndots$-competitiveness, but it has been shown as necessary for single-threshold algorithms: in \cite{Esh18} a single-threshold $0.6321\ndots$-competitive algorithm is designed, and it is also shown any such single-threshold algorithm must be tight; furthermore, in \cite{Esf17} it was also shown that if we admit point-masses in the distributions, single-threshold algorithms are $\sfrac{1}{2}$-hard. One interesting fact about RO is that $0.7451\ndots$-competitiveness can be recovered by using essentially a constant subset of the random variables \cite{Liu21}.\footnote{More specifically, given any instance $\mathbf{V}=\{V_1,\ldots,V_n\}$, for any $\varepsilon>0$ we can find an instance $\mathbf{V}'$ which is a subset of $n'=n'(\varepsilon)$ random variables from $\mathbf{V}$(note that $n'$ is independent of $n$, and thus constant with respect to the instance)
    , such that there is a stopping rule $\tau\in C^{n'}$ ensuring that the ratio of $\E V'_\tau$ over $\E\max_{i\in[n']}V'_i$ is greater than $0.7451\ndots-\varepsilon$.} Removing random variables is not the only way to push beyond the $0.7235$-hardness achieved in our work: \emph{large markets} hypothesis are also very helpful. For example, in~\cite{Abolhass17} it was shown that, if each distribution in an instance of size $n$ occurs at least $\Theta(\log n)$ times, a $0.7380$-competitive algorithm can be designed.

	\subsection{Organisation of the paper}
	In \Cref{Preliminaries}, we provide some background on optimal stopping theory and apply it to RO and Instance \ref{instance}. In \Cref{BIRO}, the general backward induction stopping rule is described informally. In \Cref{application} the explicit backward induction thresholds for Instance \ref{instance} are derived. In \Cref{acceptance} we derive the associated acceptance times, which are the starting point of our analysis. In \Cref{asymp}, we derive sharp two asymptotic estimates for these times and find their asymptotic reciprocal ordering (\Cref{accept}). 
	
	In \Cref{hardness}, we first compute the prophet's expectation (\Cref{max}) and the optimal algorithm's expectation (\Cref{optimal}) and derive sharp asymptotic estimates. Then we analyse the ratio asymptotically, which leads to an optimisation problem, solving which the upper bound of $\CR$ is derived (\Cref{mainth}).
	
	The reader can find in the Appendix the technical details of some of the proofs (\Cref{suppasymp,suppmax}); formal error analysis and reference to both the shared code used to determine the numerical value of our theoretical hardness, and the shared code that corroborates tightness of the estimate (\Cref{code}). 
	
	We included also a complementary discussion, intended solely for the reviewers, of some optional, more advanced, foundational aspects of optimal stopping theory for general finite random order processes, not required for the specific model studied in this paper (\Cref{suppreview}). In the main body these foundational aspects are addressed informally, for the sake of intuition, as customary in the prophet inequality literature.
	 
	\section{Preliminaries}
    \label{Preliminaries}
	
    In this section, we leverage classical tools from optimal stopping theory, in order to derive a general intuitive formulation of the dynamic program for RO. We then apply this formulation to Instance \ref{instance}, more formally, and derive acceptance times for the values of the instance. Finally, we asymptotically estimate the acceptance times, as the size of the instance grows.
    
	\subsection{Backward induction for random order processes}\label{BIRO}
	In a random order stopping problem, a gambler seeks to optimally stop a finite random order (stochastic) process, that is a random vector $X=(V_1, \ldots, V_{n+1})$, defined on a probability space $(\Omega,\mathcal{F},\Prob)$, whose components are permuted uniformly at random. When the order is fixed, since the (time) horizon $n+1$ is finite, standard \textit{backward induction} yields an optimal stopping rule, regardless of the dependencies between the components of the vector \cite[Theorem~3.2]{ChoRobSieg71}. Backward induction can also be naturally extended to random order processes in full generality, by making use of a filtration accounting for the randomly permuted order of arrival. 
 
    Intuitively, we construct a random order process by replacing the indices of the vector's components with those of a (uniform) random permutation of $[n+1]$, $\pi=(\pi_1,\ldots,\pi_{n+1})$, and denote it $X^\pi \defeq (V_{\pi_1},\ldots,V_{\pi_{n+1}})$. The natural filtration of this process is informally obtained by starting, as usual, with $\mathcal{F}_0 \defeq\{\emptyset, \Omega\}$, and for all $k\in[n + 1]$, letting $\mathcal{F}_k \defeq\sigma(V_{\pi_1},\ldots,V_{\pi_k})$. $\mathcal{F}_k$ denotes the $\sigma$-algebra generated by $V_{\pi_1},\ldots,V_{\pi_k}$, which is intuitively understood as the information, available to the gambler, pertaining the history of the process up to time $k$. Therefore, $\mathcal{F}_k$ encodes both the history of the random order and of the random values, within the capacity of the gambler's knowledge. A stopping time $\tau$ is a random time determined with no information other than the one provided by $\mathcal{F}_k$ (formally, $\{\tau=k\}\in\mathcal{F}_k$). It models the gambler, who can stop the process $X^\pi$, and the stopped value is denoted as the random variable $V_{\pi_\tau}$. We denote $(\Omega,\mathcal{F},\{\mathcal{F}_k\},\Prob)$ the supporting probability space for this model and let $\E$ denote the corresponding expectation and $\E_{\mathcal{F}_k}$ the expectation conditional on the history of the process up to time $k$.
    
    Recall that backward induction is a strategy for the gambler based on the following premise: \emph{take a value only if it is larger than the expected future reward}. Since at step $k$ the rewards can come only from the values $V_{\pi_k},\ldots,V_{\pi_{n+1}}$, this strategy defines backwards random rewards $\gamma_{n+1}, \ldots,\gamma_2, \gamma_1$ and stopping times $s_{n+1}, \ldots, s_2, s_1$. 
    For each $k \in [n+1]$, consider the case where the gambler rejected the first $k-1$ values. Then, the value $\E \gamma_k$ represents the expected reward that the gambler obtains using backward induction, characterised through the stopping times $s_k$, which implement the stopping rule of accepting the first value $V_{\pi_l}$ that exceeds the expected future reward (conditionally on the rejected values) $\E_{\mathcal{F}_l}\gamma_{l+1}$ for $l\geq k$. It holds that the expected reward of the stopping rule $s_1$ is optimal.\footnote{We will dive no further into the formal details of the general model, since it will not be necessary. They are included in \Cref{suppbackwardcomp} for the interested reader. There, we explicitly construct a random order filtration and state the equivalent version of \cite[Theorem~3.2]{ChoRobSieg71} for random order processes. We could not find a reference for it, and it is valuable for future work. It can be leveraged for extensions of the model, such as when allowing more general instances, random orders and dependencies.}

	\subsection{Backward induction on Instance \ref{instance}}\label{application}
	Consider RO on the following class of instances, to which it is strightforward to verify that Instance \ref{instance} belongs. 
		\begin{definition}
		Let $a,b,p\in\mathbb{R}$, $a<1<b$, be positive constants and fix $n\in\mathbb{N}$ large enough, so as to have a well defined random variable
		\[V\sim
		\begin{cases}
			n,&\text{w.p. }\:\sfrac{1}{n^2}\\
			b,&\text{w.p. }\:\sfrac{p}{n}\\
			0,&\text{w.p. }\: 1 - \sfrac{p}{n} - \sfrac{1}{n^2}.
		\end{cases}\]
		Furthermore, assume that $\log(1+pb)<p$ and that the following conditions hold:
		\[
		\frac{1+bp}{1+(b-a)p}\log\frac{1+bp}{1+(b-a)p}\leq ap \,,
		\tag*{I}
		\label{conditionI}
		\]
		\[
		(2-p)(b-a)<1 \,,
		\tag*{II}
		\label{conditionII}
		\]
		\[
		\frac{1-(2-p)(b-a)}{1+p(b-a)} < 1 + \frac{1}{p} \log\frac{1+p(b-a)}{1+pb} \,, 
		\tag*{III}
		\label{conditionIII}
		\]
		\[
		2+pb[1-p(b-a)] \ge 0 \,, 
		\tag*{IV}
		\label{conditionIV}
		\]
		\[
		\frac{bp[1+(b-a)p]}{(1+pb)\log(1+pb)}<1\,.
		\tag*{V}
		\label{conditionV}
		\]
		We define the instance $\{V_1,V_2,\ldots,V_{n+1}\}$, which consists of $n$ iid random variables distributed as $V$ and the constant $a$, that is
		\begin{equation*}
			V_i\sim\begin{cases}V,&1\leq i\leq n\\ a,&i=n+1.\end{cases}
		\end{equation*}
	\end{definition}
	To simplify our analysis, from now on, by Instance \ref{instance} we will refer to the above class, which we will need to specialise to the values $a=\aval$, $b=\bval$, $p=\pval$ only in the conclusion of \Cref{mainth}. With slight abuse of notation, we omit the reference to the permutation from $X^\pi$, thus simply denoting $X=(V_{\pi_1},\ldots,V_{\pi_{n+1}})$. All components of $X$ are independent and the process takes values in the finite state space 
    \[
        \mathcal{S}_n\defeq\{(x_1,\ldots,x_{n+1}):\;\exists!\: i\in[n+1],\, x_i=a,\:\forall\, j\neq i,\: x_j\in \{0,b,n\} \} \,.
    \]
    We can replace conditioning on the random order filtration with elementary conditioning on all the possible values of the random variables $V_1,\ldots,V_{n+1}$ in all the possible arrangements, due to $\mathcal{F}_k$ being generated by the singletons $\{ (X_1 = x_1, \ldots, X_k = x_k ) : \: x \in \mathcal{S}_n \}$ for all $k\in[n+1]$.
    
    We now state more formally backward induction in the discrete setting. Let $C_k^{n+1}$ be the set of stopping rules $\tau$ that never stop before time $k$, that is such that $k\leq \tau\leq n+1$. Note that we previously denoted $C^{n+1}\defeq C_1^{n+1}$. For ease of notation, in this section we omit the reference to the horizon and simply denote these classes as $C_k$. For every $x\in\mathcal{S}_n$, define backwards $\gamma_{n+1}, \gamma_n,\ldots, \gamma_1$ as functions of $x$, by setting $\gamma_{n+1}(x_1,\ldots,x_{n+1})\defeq x_{n+1}$ and for all $l=n,\ldots,1$, 
		\begin{align*}
			\gamma_l(x_1,\ldots,x_l)&\defeq\max\{x_l,\overline{\gamma}_l(x_1,\ldots,x_l)\},\\
			\overline{\gamma}_l(x_1,\ldots,x_l)&\defeq\E[\gamma_{l+1}(x_1,\ldots,x_l,X_{l+1})|X_1=x_1,\ldots,X_l=x_l].
		\end{align*}
		For each $k=n+1,\ldots,1$, let 
		\[s_k\defeq\inf\{l\geq k:\; \gamma_l(x_1,\ldots,x_l)=x_l\}\]
        and denote $T\defeq s_1$.
	\begin{theorem}\label{backwarddiscrete}
		Consider $X\defeq(V_{\pi_1},\ldots,V_{\pi_{n+1}})$, the backward induction values $\gamma_{n+1},\gamma_n, \ldots,\gamma_1$ and the stopping times $s_{n+1}, s_n,\ldots, s_1$ as previously defined. Then for any $k \in [n + 1]$, $s_k \in C_k$ and for any realization $(x_1,\ldots,x_k)$, the value $\gamma_k(x_1,\ldots,x_k)$ is the optimal value the gambler can obtain from time $k$. More formally, for all stopping rules $\tau\in C_k$ we have that
		\[
            \E ( X_{s_k} \mid X_1 = x_1, \ldots, X_k = x_k ) = \gamma_k(x_1,\ldots,x_k) \geq \E (X_\tau|X_1=x_1,\ldots,X_k=x_k) \,.
        \]
        Thus for all stopping rules $\tau\in C_k$,
		\[
            \E X_{s_k}=\E\gamma_k(X_1,\ldots,X_k)\geq \E X_\tau \,.
        \]
        In particular, for all stopping rules $\tau\in C_1$,
        \[
            \E V_{\pi_{T}}\defeq\E X_T=\E\gamma_1(X_1)\geq \E X_\tau\defeq \E V_{\pi_\tau} \,.
        \] 
        Thus $\E\gamma_1(X_1)=\sup_{\tau\in C_1}\E V_{\pi_\tau}$.
	\end{theorem}
	The nested maxima in the definition of the future expectations $\overline{\gamma}_k(x_1,\ldots,x_k)$ imply monotonicity of these thresholds.
	\begin{remark}
		\label{mono}
		For every $x \in \mathcal{S}_n$ fixed, we have that for every $k \in [n-1]$, $\overline{\gamma}_k(x_1,\ldots,x_k)\geq\overline{\gamma}_{k+1}(x_1,\ldots,x_{k+1})$.
	\end{remark}
    Recall that Instance \ref{instance} has the following property: with undisclosed RO the stopping rule, which at any time $k\in[n+1]$ is not given information about the past values of the random ordering $\pi_1,\ldots,\pi_k$, will nonetheless be able to infer, step by step, whether the values probed came from $V$ or not, by simply observing if $a$ has already been probed or not. Moreover, from the fact that
		\[\Prob(\pi_{k+1}=n+1|X_1=x_1,\ldots,X_k=x_k)=\begin{cases}0,& \exists\,1\leq i\leq k:\,x_i=a\\\frac{1}{n-k+1},&\text{otherwise,}\end{cases}\]
    since $\{\pi_{k+1}=n+1\}=\{X_{k+1}=a\}$, we can compute the quantities in \Cref{backwarddiscrete} explicitly.
    
	\begin{remark}\label{backward}
		For Instance \ref{instance}, $\overline{\gamma}_n(x_1,\ldots,x_n)=\E V= \frac{1+bp}{n}$, if there exists $1\leq i\leq n$ such that $x_i= a$, otherwise it equals $a$. 
		For all $1\leq k< n$, 
		\[\overline{\gamma}_k(x_1,\ldots,x_k)=\begin{cases}\E\gamma_{k+1}(x_1,\ldots,x_k,V),& \exists\,1\leq i\leq k :\,x_i=a\\\frac{\gamma_{k+1}(x_1,\ldots,x_k,a)}{n-k+1}+\left(1-\frac{1}{n-k+1}\right)\E\gamma_{k+1}(x_1,\ldots,x_k,V),&\text{otherwise,}\end{cases}
		\]
		where $\:\E\gamma_{k+1}(x_1,\ldots,x_k,V)=\frac{\gamma_{k+1}(x_1,\ldots,x_k,n)}{n^2}+\frac{p\gamma_{k+1}(x_1,\ldots,x_k,b)}{n}+\gamma_{k+1}(x_1,\ldots,x_k,0)\left(1-\frac{p}{n}-\frac{1}{n^2}\right)$.
	\end{remark}

    
	\subsection{Acceptance times}\label{acceptance}
	It is well known that there is no loss of generality in \Cref{backwarddiscrete} considering only nonrandomised stopping rules and that computing directly a closed form for the expected value of the optimal algorithm $\E V_{\pi_T}$ is usually infeasible.\footnote{The unfamiliar reader is referred to \Cref{suppbackward} for foundational details on this aspects.} Nonetheless, \Cref{backwarddiscrete,backward} enable us to infer key asymptotic features of $T$, yielding nearly tight estimates of $\E V_{\pi_T}$. We will characterise $T$ in terms of \emph{acceptance times}: for each value $0$, $a$, $b$, $n$, there is a time before which the value is never accepted and such that, from that time onward, the value would always be accepted.\footnote{\textit{Accepting} a value means that, if the value is probed, the algorithm stops and the value is the reward. Acceptance times are well defined by \Cref{mono}.} 
	The acceptance times for the values $0$ and $n$ are trivial: if $n$ is probed, the optimal algorithm always stops; if $0$ is probed, it does not stop unless it is in the last step.
	\begin{remark}\label{trivial}
    	For all $x\in \mathcal{S}_n$ and $k\in[n]$, we have $0<\overline{\gamma}_k(x_1,\ldots,x_k)\leq n$.
    	Thus the acceptance time for $n$ is $1$ and the acceptance time for $0$ is $n+1$. 
	\end{remark}

    The acceptance times for the values $a$ and $b$ are nontrivial. However, the earliest time that $b$ may be accepted is random
    . Therefore, we define two acceptance times for the value $b$, depending on whether the value $a$ has already been probed or not. 
    To define the acceptance times for $a$ and $b$, we exploit the fact that the thresholds are deterministic, given the information regarding whether $a$ has already been revealed or not. For $i \in [n]$, denote the event that $a$ is probed at time $i$ by $\Omega_i \defeq \{ \pi_i = n + 1 \} = \{ X_i = a \}$, and denote the corresponding conditional expectation $\E_i(\cdot) \defeq \E(\,\cdot \mid \Omega_i)$ and conditional probability $\Prob_i(\cdot)\defeq \Prob(\,\cdot \mid \Omega_i)$. 
    For every $k \in [n + 1]$, denote the event that the value $a$ is probed by time $k$ as
    $ \Pi_k \defeq \bigcup_{i\in[k]} \Omega_i  = \{ \exists\, 1 \leq i \leq k : \,\pi_i = n + 1 \}$.
    An immediate consequence of \Cref{backward} is the following. 
	\begin{remark}\label{conditionalindep}
		For every $k\in [n]$ and for any $x\in\mathcal{S}_n$, $\overline{\gamma}_k(x_1,\ldots, x_k)$ is determined by the the information regarding $\Pi_k$, that is by the time step $k$ and whether there is $i\in[k]$ such that $x_i=a$ or not.
	\end{remark}

   	By \Cref{conditionalindep}, for Instance \ref{instance}, the future rewards can be computed solely based on the time stage $k$ and whether the value $a$ has been probed or not. For every $k \in [n]$ and $x\in\mathcal{S}_n$, they are characterised respectively by
	\begin{align*}
        \phi_k^{n+1}
            &\defeq \E( V_{\pi_T} \mid T > k, \Pi_k)
            = \overline{\gamma}_k(x_1, \ldots, x_i = a, \ldots, x_k) \\
        \bar{\phi}_k^{n+1}
            &\defeq \E( V_{\pi_T} \mid T > k, \Pi_k^c)
            =\overline{\gamma}_k(x_1 \neq a, \ldots, x_k \neq a) \,.
	\end{align*} 
	We omit the time horizon, so that  we denote the expected future reward when $a$ has been already probed at time $k$ by $\phi_k$ and the expected future reward when $a$ has not been already probed yet at time $k$ by $\bar{\phi}_k$. Furthermore, we will denote $x\vee y\defeq\max\{x,y\}$.
	\begin{remark}
    \label{backwardphi}
        \Cref{backward} is equivalent to $\phi_n = \E V= \frac{1+bp}{n}$, $\bar{\phi}_n = a$, and for all $k \in [n-1]$,
		\begin{align*}
			\phi_k 
                &= \E(V\vee\phi_{k+1}) \,, \\
			\bar{\phi}_k 
                &= \frac{a\vee\phi_{k+1}}{n+1-k}
                    +\left(1-\frac{1}{n+1-k}\right) \E(V\vee\bar{\phi}_{k+1}) \,.
		\end{align*}
	\end{remark}
 	Finally, we define formally the acceptance times in terms of $\phi_k$ and $\bar{\phi}_k$.
	\begin{definition}\label{defj}
		Denote the earliest (deterministic) time that the value $a$ would be accepted, if probed, by $j_n \defeq \inf \{ k \in [n+1] : \: a \geq \phi_k \}$.
	\end{definition}
 
    
    
	\begin{definition}\label{defk}
        Denote the earliest (deterministic) time that the value $b$ would be accepted if probed, given that the value $a$ has already been probed, by $k_n \defeq \inf \{ k \in [n + 1] :\: b \geq \phi_k \}$.
	\end{definition}


	\begin{definition}
    \label{defbark}
        Denote the earliest (deterministic) time that the value $b$ would be accepted if probed, given that the value $a$ has not been probed, by $\bar{k}_n \defeq \inf \{ k \in [n + 1] :\: b \geq \bar{\phi}_k \}$.
	\end{definition}
	
    Note that, for all large enough $n$, we have that $j_n,\,k_n,\,\bar{k}_n \le n$ since $\phi_n \leq a<b$ and, at time $n$, not having probed the value $a$ implies that $x_{n+1} = a=\bar{\phi}_n$.

    \subsection{Asymptotic estimates for the acceptance times}
    \label{asymp}
    
    In this section we derive sharp asymptotic estimates on $j_n$, $k_n$, and the relative order of $j_n$, $k_n$ and $\bar{k}_n$. The proof is quite technical and involved, details are provided in \Cref{suppasymp}. Here we discuss its high-level ideas, and provide a graphical representation for the hard instance in \Cref{prophet}. All logarithms refer to the natural logarithm. 
    
    It is important to keep in mind, from the previous section, that: for all times after $j_n$, the optimal stopping rule will always accept the value $a$ if probed; if the value $a$ has already been probed, then, for all times after $k_n$, the optimal stopping rule will accept the value $b$ if probed; if the value $a$ has not yet been probed, then, for all times after $\bar{k}_n$, the optimal stopping rule will accept the value $b$ if probed.
	\begin{lemma}\label{accept}
		As $n\longrightarrow\infty$, the optimal stopping rule for Instance \ref{instance} is such that we have the following.
		\begin{enumerate}[a)]
			\item$ j_n \sim n\left(1+\frac{1}{p}\log\frac{1+(b-a)p}{1+bp}\right).$
			\item$k_n\sim n\left(1+\frac{1}{p}\log\frac{1}{1+bp}\right).$
			\item $k_n \le j_n$. Informally, the gambler accepts the value $a$ later than the value $b$ after seeing the value $a$.
			\item$k_n \le \bar{k}_n$. Informally, the gambler accepts the value $b$ not having seen the value $a$ later than the value $b$ after seeing the value $a$.
			\item$\bar{k}_n \le j_n$. Informally, the gambler accepts 
            the value $a$
            later than 
            the value $b$ not having seen the value $a$.
		\end{enumerate}
	\end{lemma}
    \begin{proof}[Idea of the proof]
        \begin{description}[style=unboxed, leftmargin=0cm]
        \item ~
            \begin{enumerate}[a)]
			\item By iterating the formulas in \Cref{backwardphi} we obtain a close expression for $\{\phi_k\}$,  for all $ k\leq n$:
			\begin{equation*}
				\phi_k=\frac{1+bp}{p}\left[1-\left(1-\frac{p}{n}\right)^{n-k+1}\right]+\mathcal{O}\left(\frac{1}{n}\right),
			\end{equation*} 
			which can be turned into a sharp asymptotics for $j_n$ as $n\longrightarrow\infty$, since this is the smallest $k$ such that $a\geq\phi_k$. Given the above expression, this equation is in fact recast as 
			\begin{equation*}
				a\geq \frac{1+bp}{p}\left[1-\left(1-\frac{p}{n}\right)^{n-i+1}\right]+\mathcal{O}\left(\frac{1}{n}\right).
			\end{equation*}
			Asymptotically, this yields the claim, since 
				\[j_n=\bigg\lceil n\left[1+\frac{1}{p}\log\left(\frac{1+(b-a)p}{1+bp}\right)\right]+\mathcal{O}\left(1\right)\bigg\rceil.\]
				
			\item Since $k_n$ is the smallest $k$ such that $b\geq\phi_k$, the asymptotics for $k_n$ is obtained similarly to the one for $j_n$ in \Cref{accept} (a).
			
			\item Follows from the fact that, as $n\longrightarrow\infty$, the limit of $\sfrac{j_n}{n}$ obtained in \Cref{accept} (a) is larger than the limit of $\sfrac{k_n}{n}$ obtained in \Cref{accept} (b) since $a<b$.
			
			\item Assume by contradiction that $\bar{k}_n<k_n$ infinitely often as $n\longrightarrow\infty$. For all such $n$ this implies that $\bar{\phi}_{k_n-1}\leq b$. We will reach a contradiction with this fact as follows. First, through \Cref{backwardphi}, by induction we derive the iterative lower bound on $\bar{\phi}_k$ for all $k \ge k_n-1$, for all such $n$, ensuring that
			\begin{equation*}
				\bar{\phi}_k\ge a+\left(\frac{1+(b-a)p}{n}-\frac{a}{n^2}\right)\sum_{j=0}^{n-k-1}\frac{n-k-j}{n+1-k}\left(1-\frac{p}{n}-\frac{1}{n^2}\right)^j.
			\end{equation*} 
			For $k=k_n-1$ we apply \Cref{accept} (b) to the above, which yields, upon asymptotically estimating the summation terms, to establishing that
			\begin{align*}
				\bar{\phi}_{k_n-1}\geq b+\frac{1}{p}-\frac{b[1+(b-a)p]}{(1+pb)\log(1+pb)}+\smallO(1).
			\end{align*}
			By Condition \ref{conditionV},  this implies $\bar{\phi}_{k_n}>b$ as $n\longrightarrow\infty$, which yields the contradiction sought.
			
			\item  Assume by contradiction that $j_n < \bar{k}_n$ infinitely often as $n \longrightarrow \infty$. We first derive, through \Cref{backwardphi}, an iterative upper bound on $\{\bar{\phi}_k\}$ for all $k \ge j_n$, for all such $n$. The derivation is a little involved. Since by assumption for all $\bar{k}_n\le k\le n$, $\bar{\phi}_k\leq b$ and $\phi_k\leq a$, and for all $j_n\le k<\bar{k}_n$, $\bar{\phi}_k\geq b$ and $\phi_k\geq a$, by an induction argument split into the two segments, for all $j_n\leq k\leq n-1$, we have
			\begin{equation*}
				\bar{\phi}_{k}\leq a+\frac{n-k}{2}\frac{1+(b-a)p}{n}.
			\end{equation*} 
			Then we consider the earliest time that $b$ is greater or equal to this upper bound, and denote it as $k_n^*$. The definition can be recast as 
			\[k_n^*=\left\lceil n\frac{1-(2-p)(b-a)}{1+p(b-a)} \right\rceil\]and thus, as $n\longrightarrow\infty$, it holds that
            \begin{itemize}[noitemsep]
            \item $\bar{k}_n \leq k_n^*$. 
            \item $k_n^*\sim \gamma n$, where by Condition \ref{conditionII} and $a<b$,
            \[
           \gamma = \frac{1-(2-p)(b-a)}{1+p(b-a)}\in(0,1).
            \]
            \end{itemize} 
            We show that by \Cref{accept} (a) this value of $\gamma$ implies that $k_n^* < j_n$ for all sufficiently large $n$. This is a contradiction with the hypothesis that $j_n < \bar{k}_n \leq k^*_n$ for infinitely many values of $n$.
		\end{enumerate}
        \end{description}
    \end{proof}
    
    \begin{figure}
        \centering
        \begin{subfigure}[b]{0.49\textwidth}
             \centering
             \includegraphics[width=\textwidth]{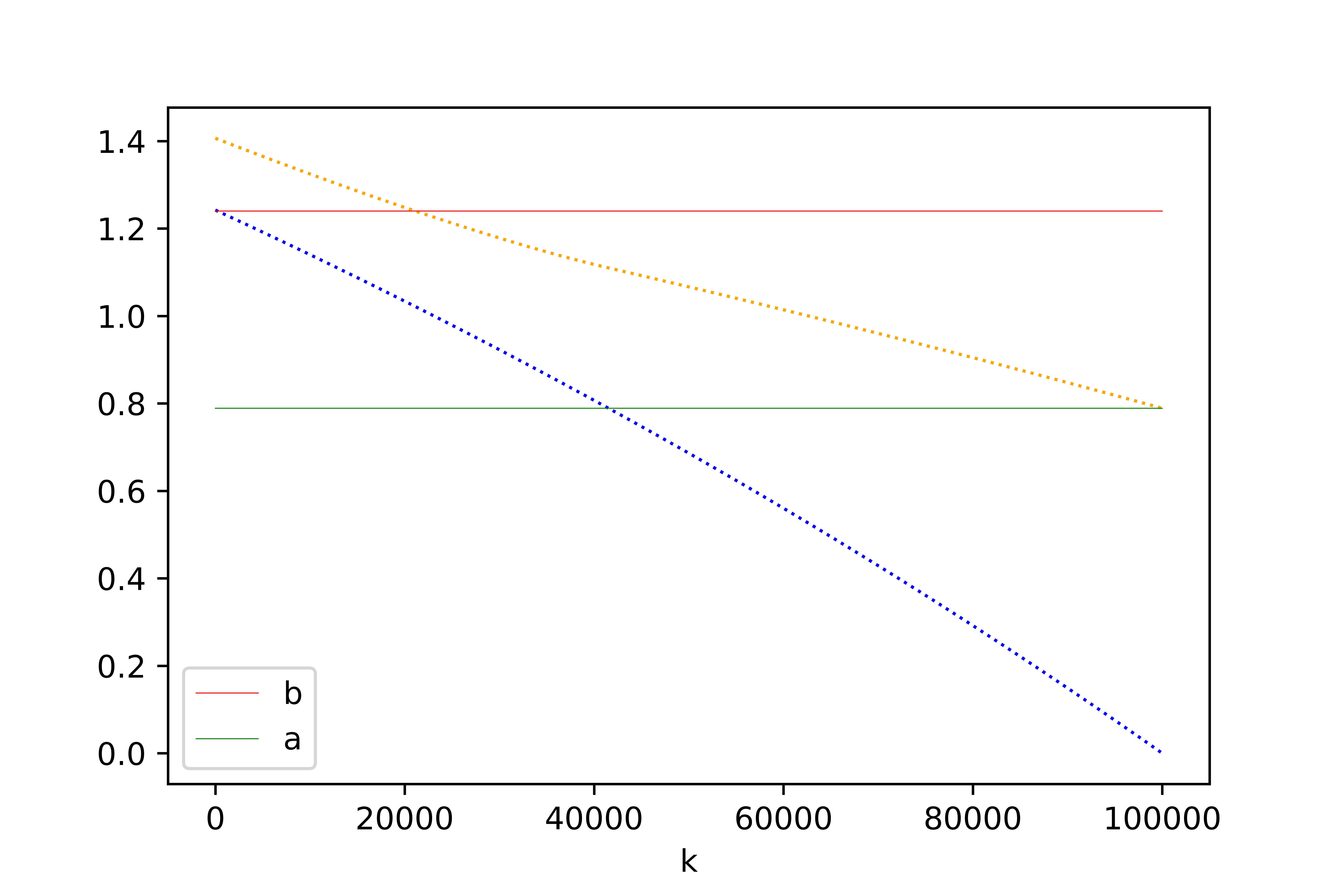}
             \caption{Future rewards $\{\phi_k\}$ (blue) and $\{\bar{\phi}_k\}$ (amber)}
             \label{fig: prophet full}
         \end{subfigure}
        \begin{subfigure}[b]{0.49\textwidth}
             \centering
            \includegraphics[width=\textwidth]{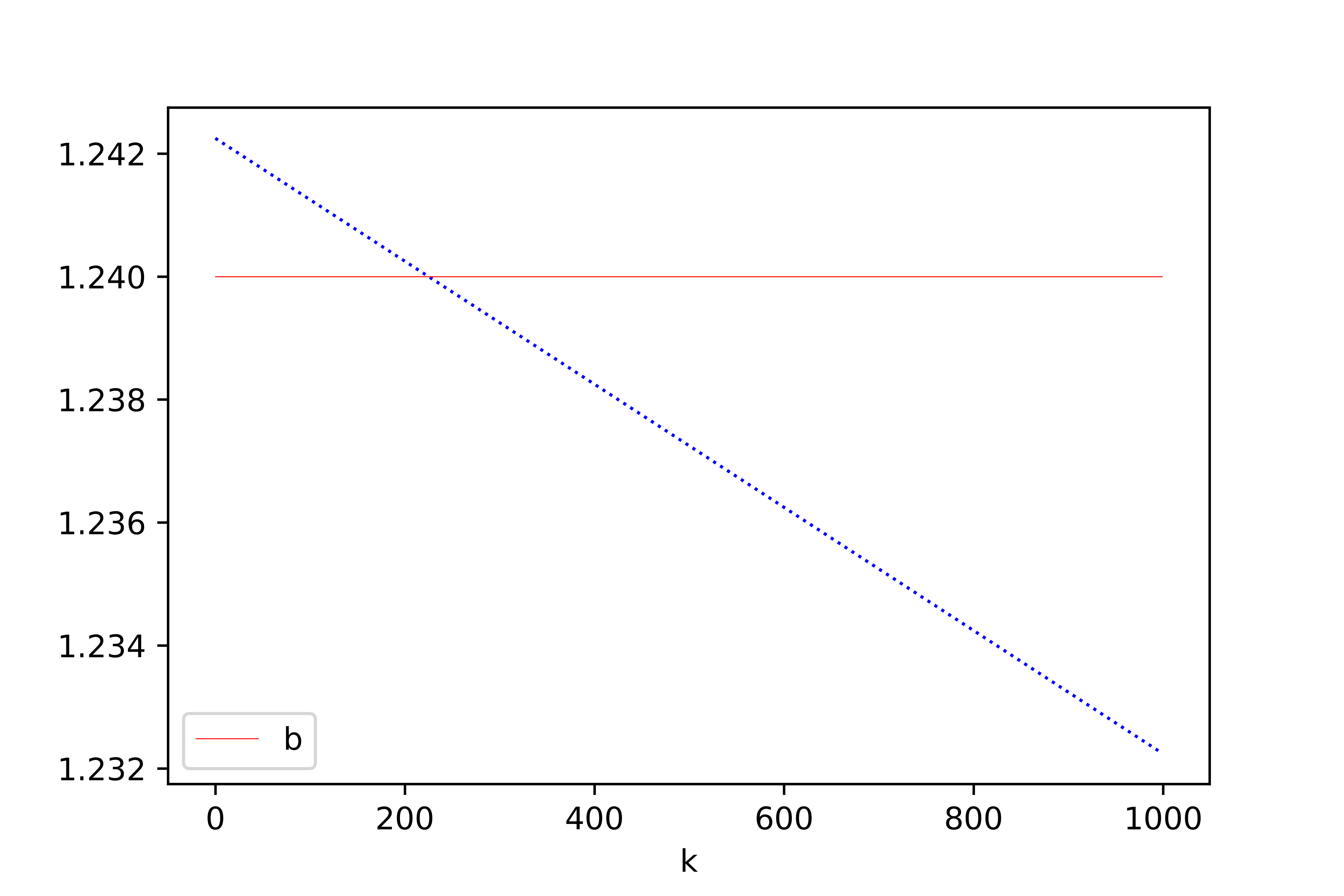}
             \caption{Zoom in at the intersection $\phi_k \approx b$}
             \label{fig: prophet zoom}
         \end{subfigure}
        \caption{Simulation of the dynamic program (reference to the code shared is in \Cref{code}) for Instance \ref{instance} with $a=\aval$, $b=\bval$, $p=\pval$, $n=10^6$. 
        \Cref{fig: prophet full} shows the sequences $\{\phi_k\}$ (blue), $\{\bar{\phi}_k\}$ (amber), and the values $a$ (green), $b$ (red).  
        \Cref{fig: prophet zoom} shows a zoom on the intersection where $\phi_k \approx b$. 
        Informally, the abscissa of the intercept of the blue dotted curve with the red line corresponds to the smallest acceptance time $k_{10^6}\approx2253$; 
        the abscissa of the intercept of the amber dotted curve with the red line is the second largest acceptance time $\bar{k}_{10^6}\approx211231$; the abscissa of the intercept of the blue dotted curve with the green line is the largest acceptance time $j_{10^6}\approx415187$, as per \Cref{accept} (c, d, e).}
        \label{prophet}
    \end{figure}
    
	\section{Random Order is $\CR$-hard}\label{hardness}
	We start by calculating a sharp asymptotic estimate for the prophet's expectation. The elementary computation is provided in \Cref{suppmax}.
	\begin{lemma}\label{max} 
		For Instance \ref{instance} as $n\longrightarrow\infty$,
		$\E\max_{i\in[n+1]}V_i=1+b(1-e^{-p})+ae^{-p}+\mathcal{O}\left(\sfrac{1}{n}\right)$.
	\end{lemma}
	Next we derive sharp asymptotic estimates for the expectation of the optimal algorithm. The details of the proof are provided in \Cref{suppmax}.
	\begin{proposition}\label{optimal}
		For Instance \ref{instance} as $n\longrightarrow\infty$,
		$\E V_{\pi_T}= q_{a,b,p}(\lambda_n,\mu_n,\nu_n)+\mathcal{O}\left(\sfrac{1}{n}\right)$,
		where $\lambda_n \defeq \sfrac{j_n}{n}$, $\mu_n \defeq \sfrac{k_n}{n}$, $\nu_n \defeq \sfrac{\bar{k}_n}{n}$ and $q_{a,b,p}(\lambda,\mu,\nu)$ is a multivariate exponential quadratic in the variables $\lambda$, $\mu$, $\nu$ and parameters $a$, $b$, $p$ defined as
		\begin{align*}
			q_{a,b,p}(\lambda,\mu,\nu)&=\frac{\mu^2}{2}-\frac{\nu^2}{2}+\nu+\frac{1}{p}+b+\left(\frac{1}{p}-\mu\right)\left(\frac{1}{p}+b\right)e^{p(\mu-1)}\\&+\left[\left(\frac{1}{p}+b\right)(\nu-\lambda)-\frac{a}{p}\right]e^{p(\nu-1)}-\frac{1}{p}\left(\frac{1}{p}+b-a\right)e^{p(\nu-\lambda)}.
		\end{align*}
	\end{proposition}
        \begin{proof}[Idea of the proof]
        In order to obtain sharp estimates of $\E V_{\pi_T}$ we rely on the eventual ordering $k_n\leq \bar{k}_n\leq j_n$ and the corresponding asymptotics derived in \Cref{accept}, exploiting the partitioning $\{\Omega_i\}$ of the probability space, via the law of total expectation. The role of \Cref{accept} in this computation can be appreciated from the following example. Consider $i<j_n$ for instance: conditionally on $\Omega_i$, the optimal stopping rule does not stop when $a$ or $0$ are probed (except for $0$, in the last step), but it stops the first time $n$ is probed, or, at certain times, when $b$ is probed. To determine when $b$ is accepted or not, knowledge of the relative position of $\bar{k}_n$ with respect to $k_n$ and $j_n$ is needed. Thanks to \Cref{accept}, on top knowing exact asymptotics, we avoid a lengthy case analysis, since only one ordering is possible: $k_n\leq \bar{k}_n\leq j_n$. We start by finding the distribution of $(V_{\pi_T}|\Omega_i)$ for all $i\in[n+1]$ (the abuse of notation denotes the conditional distribution of $V_{\pi_T}$ given $\Omega_i$), so as to determine, for all $n$ large enough, $\E_iV_{\pi_T}$ for all $i\in[n+1]$. Recall that, conditionally on $\Omega_i$, at every step of the process other than the $i$th, independently, the algorithm could probe $n$, $b$ and $0$, with probabilities $\sfrac{1}{n^2}$, $\sfrac{p}{n}$ and $1-\sfrac{p}{n}-\sfrac{1}{n^2}$ respectively. While $n$ is always accepted, $0$ is never accepted, until the end, by the optimal algorithm. As $n\longrightarrow\infty$,
        \begin{enumerate}[i), noitemsep]
			\item If $i< k_n$, then we have 
			\begin{align*}
				\E_iV_{\pi_T}&=n\left[\frac{k_n}{n^2}+\mathcal{O}\left(\frac{1}{n^2}\right)\right]+\left[1-\frac{k_n}{n^2}+\mathcal{O}\left(\frac{1}{n^2}\right)\right]\left[\left(\frac{1}{p}+b\right)\left(1-e^{p\left(\frac{k_n}{n}-1\right)}\right)+\mathcal{O}\left(\frac{1}{n}\right)\right]\\&=\frac{k_n}{n}+\left(\frac{1}{p}+b\right)\left(1-e^{p\left(\frac{k_n}{n}-1\right)}\right)+\mathcal{O}\left(\frac{1}{n}\right).
			\end{align*}
             \item If $k_n \leq i < \bar{k}_n$, then we have $\E_iV_{\pi_T}=\frac{i}{n}+\left(\frac{1}{p}+b\right)\left(1-e^{p\left(\frac{i}{n}-1\right)}\right)+\mathcal{O}\left(\frac{1}{n}\right)$.
            \item If $\bar{k}_n\leq i<j_n$, then we have $\E_iV_{\pi_T}=\frac{\bar{k}_n}{n}+\left(\frac{1}{p}+b\right)\left(1-e^{p\left(\frac{\bar{k}_n}{n}-1\right)}\right)+\mathcal{O}\left(\frac{1}{n}\right)$.
            \item If $i\geq j_n$, then we have 
            \begin{align*}
					\E_iV_{\pi_T}&=n\left[\frac{\bar{k}_n}{n^2}+\mathcal{O}\left(\frac{1}{n^2}\right)\right]+\left[1-\frac{\bar{k}_n}{n^2}+\mathcal{O}\left(\frac{1}{n^2}\right)\right]\left[\left(\frac{1}{p}+b\right)\left(1-e^{p\left(\frac{\bar{k}_n}{n}-\frac{i}{n}\right)}\right)\right]\\&+a\left(1-\frac{\bar{k}_n}{n^2}+\mathcal{O}\left(\frac{1}{n^2}\right)\right)\left(e^{p\left(\frac{\bar{k}_n}{n}-\frac{i}{n}\right)}+\mathcal{O}\left(\frac{1}{n}\right)\right)+\mathcal{O}\left(\frac{1}{n}\right)\\&=\frac{\bar{k}_n}{n}+\left(\frac{1}{p}+b\right)\left(1-e^{p\left(\frac{\bar{k}_n}{n}-\frac{i}{n}\right)}\right)+ae^{p\left(\frac{\bar{k}_n}{n}-\frac{i}{n}\right)}+\mathcal{O}\left(\frac{1}{n}\right).
            \end{align*}
		\end{enumerate}
			Having exhausted all cases, by the law of total expectation we can compute, by adopting empty sum convention,
			\begin{align*}
				\E V_{\pi_T}&=\frac{1}{n+1}\sum_{i=1}^{k_n-1}\E_i V_{\pi_T}+\frac{1}{n+1}\sum_{i=k_n}^{\bar{k}_n-1}\E_i V_{\pi_T}+\frac{1}{n+1}\sum_{i=\bar{k}_n}^{j_n-1}\E_i V_{\pi_T}+\frac{1}{n+1}\sum_{i=j_n}^{n+1}\E_i V_{\pi_T}\\&\defeq S_{1,k_n-1}+S_{k_n,\bar{k}_n-1}+S_{\bar{k}_n,j_n-1}+S_{j_n,n+1}.
			\end{align*}
   By the previous estimates, we obtain that, as $n\longrightarrow\infty$,
   \begin{enumerate}[i), noitemsep]
		\item $S_{1,k_n-1}=\mu_n^2+\mu_n\left(\frac{1}{p}+b\right)\left(1-e^{p\left(\mu_n-1\right)}\right)+\mathcal{O}\left(\frac{1}{n}\right)$.
        \item $S_{k_n,\bar{k}_n-1}=\frac{\nu_n^2}{2}-\frac{\mu_n^2}{2}+\left(\frac{1}{p}+b\right)(\nu_n-\mu_n)-\frac{e^{-p}}{p}\left(\frac{1}{p}+b\right)(e^{p\nu_n}-e^{p\mu_n})+\mathcal{O}\left(\frac{1}{n}\right)$.
        \item $S_{\bar{k}_n,j_n-1}=-\nu_n^2+\lambda_n\nu_n+(\lambda_n-\nu_n)\left(\frac{1}{p}+b\right)\left(1-e^{p\left(\nu_n-1\right)}\right)+ \mathcal{O}\left(\frac{1}{n}\right)$.
		\item $S_{j_n,n+1}=-\lambda_n\nu_n+\nu_n+\left(\frac{1}{p}+b\right)(1-\lambda_n)-\frac{1}{p}\left(\frac{1}{p}+b-a\right)(e^{-p(\lambda_n-\nu_n)}-e^{-p(1-\nu_n)})+\mathcal{O}\left(\frac{1}{n}\right)$.
        \end{enumerate}
        Putting these asymptotic estimates together yields the claim, upon a few rearrangements and cancellations.
    \end{proof}
	Finally, we prove our main result: the state-of-the-art $\CR$-hardness of RO and the separation of RO from OS.
	\begin{proof}[Proof of \Cref{mainth}]
		By \Cref{optimal,accept} (a, b), letting \[\lambda_*\defeq 1+\frac{1}{p}\log\frac{1+(b-a)p}{1+bp},\quad \mu_*\defeq 1+\frac{1}{p}\log\frac{1}{1+bp},\] it follows that $\E V_{\pi_T}= q_{a,b,p}(\lambda_*,\mu_*,\nu_n)+\smallO(1)$. For simplicity we will omit the reference to the parameters $a$, $b$, $p$, $\lambda_*$ and $\mu_*$ in the notation of the exponential quadratic. The dependence on $n$ is implicit in the notation: recall that in $\E V_{\pi_T}$, $T\in C^{n+1}$. Since as $n\longrightarrow\infty$, $\mu_n\longrightarrow\mu_*$, $\mu_n\leq\nu_n\leq\lambda_n$ and $\lambda_n\longrightarrow\lambda_*$, $\E V_{\pi_T}\leq \max_{\nu\in[\mu_*,\lambda_*]}q(\nu)+\smallO(1)$. By \Cref{max}, as $n\longrightarrow\infty$
		\[\frac{\E V_{\pi_T}}{\E\max_{i\in[n+1]}V_i}\leq \frac{\max_{\nu\in[\mu_*,\lambda_*]}q(\nu)}{1+b(1-e^{-p})+ae^{-p}}+\smallO(1).\]
		Denote as $m(a,b,p)\defeq\max_{\nu\in[\mu_*,\lambda_*]}q(\nu)$.
		Then 	
		\begin{equation}\label{limsup}
			\limsup_{n\longrightarrow\infty}\frac{\E V_{\pi_T}}{\max_{i\in[n+1]}V_i}\leq \frac{m(a,b,p)}{1+b(1-e^{-p})+ae^{-p}}\defeq M(a,b,p).
		\end{equation}
  		We start by computing $m(a,b,p)$. Then the maximisation is performed on $\mu_*\leq \nu\leq \lambda_*$. Note that
		\[q'(\nu)=1-\nu+\left[\left(\frac{1}{p}+b\right)+(1+bp)(\nu-\lambda_*)-a\right]e^{p(\nu-1)}-\left(\frac{1}{p}+b-a\right)e^{p(\nu-\lambda_*)},\]
		and since 
		$\left(p^{-1}+b-a\right)e^{p(\nu-\lambda_*)}=\left(p^{-1}+b\right)e^{p(\nu-1)}$, it follows that 
		\begin{align*}
			q'(\nu)&=1-\nu+[(1+bp)(\nu-\lambda_*)-a] e^{p(\nu-1)},\\
			q''(\nu)&=-1+\left\lbrace 1+p[b-a+(1+bp)(\nu-\lambda_*)]\right\rbrace e^{p(\nu-1)},\\
			q'''(\nu)&=p\{2+p[2b-a+(1+bp)(\nu-\lambda_*)]\} e^{p(\nu-1)}. 
		\end{align*}
		Note that $q'''(\nu)>0$, because $0\geq\nu-\lambda_*\geq\mu_*-\lambda_*=-p^{-1}\log[1+p(b-a)]$ for all $\mu_*\leq\nu\leq\lambda_*$ implies that
		\[2+p[2b-a+(1+bp)(\nu-\lambda_*)]\geq2+p\left[2b-a-\frac{1+bp}{p}\log[1+p(b-a)\right]\geq2+pb[1-p(b-a)]>0,\]
		where we used $\log(1+x)\leq x$ in the second last inequality and Condition \ref{conditionIV} in the last inequality. It follows that $q'(\nu)$ is convex in the interval considered for maximisation. Since 
		\[q'(\mu_*)=\frac{1}{p}\log\left(\frac{1+bp}{1+(b-a)p}\right)-\frac{a}{1+bp},\quad q'(\lambda_*)=\frac{1}{p}\log\left(\frac{1+bp}{1+(b-a)p}\right)-a\frac{1+(b-a)p}{1+bp}\,,\]
		we can conclude that $q'(\mu_*)>q'(\lambda_*)$. Since by Condition \ref{conditionI} we have $q'(\lambda_*)\leq0$, by convexity we have only two possibilities (depending on $a$, $b$, $p$).
		\begin{itemize}
			\item $q'(\mu_*)\leq0$: then $q'(\nu)\leq0$ for all $\mu_*\leq\nu\leq\lambda_*$; in this case $q(\nu)$ is nonincreasing and therefore $q(\mu_*)=\max_{\nu\in[\mu_*,\lambda_*]}q(\nu)$, hence
			$ m(a,b,p)=q(\mu_*)$, which can be computed directly. Plugging the value of $m(a,b,p)$ into the definition of $M(a,b,p)$ in \Cref{limsup} yields
			$M(a,b,p)=\frac{q(\mu_*)}{1+b(1-e^{-p})+ae^{-p}}\,$;
			\item $q'(\mu_*)>0$: then $q'(\nu_*)=0$ for some $\mu_*<\nu_*<\lambda_*$ and $q(\nu)$ is not monotonic: it increases until it reaches its maximum at $\nu_*$ and then decreases. The decimal approximation of $\nu_*$ is computed numerically, since it is a nontrivial zero of the exponential linear polynomial $q'(\nu)$. This comes with no loss of rigour, since $q'(\nu)$ is smooth and has opposite signs at the ends of the interval $[\mu_*,\lambda_*]$, ensuring that we can exploit the \emph{bisection} algorithm to determine $\nu_*$, meaning that it can be found with arbitrary accuracy, provided enough computational power. Upon computing $q(\nu_*)=\max_{\nu\in[\mu_*,\lambda_*]}q(\nu)$ directly, we obtain $m(a,b,p)=q(\nu_*)$, plugging which into the definition of $M(a,b,p)$ in \Cref{limsup}, yields
			$M(a,b,p)=\frac{q(\nu_*)}{1+b(1-e^{-p})+ae^{-p}}\,$.
  		\begin{figure}
                \centering
			\pgfplotsset{scaled y ticks=false}
			\begin{tikzpicture}[scale =0.7]
				\begin{axis}[
					xmin=0, xmax=.5,
					ymin=-.05, ymax=.1,
					axis x line=middle,
					axis y line=middle,
					xlabel={$\nu$},
					ylabel={},
					xtick={.015, .211, .415, .5},
					xticklabels={$\mu_*$, $\nu_*$, $\lambda_*$, $0.5$},
					ytick={-.05,.1}, 
					yticklabels={$-0.05$, $0.1$}
					]
					\addplot[no marks,blue] expression[domain=0:.5,samples=100]{1-x+((1+1.24*.421)*(x-.415)-.789)*exp(.421*(x-1))} node[pos=.1,anchor=south west]{$q'(\nu)$}; 
				\end{axis}
			\end{tikzpicture}
                \caption{Root $\nu_*$ of $q'(\nu)$ in $[\mu_*,\lambda_*]$}
			\label{opt_sol}
		\end{figure}
            \end{itemize}

		From now on we fix $a=\aval$, $b=\bval$, $p=\pval$, thus specialising the class defined in \Cref{application} to Instance \ref{instance}, introduced in \Cref{hardinstance}. For these values, we compute directly that we are in the second nonmonotonic scenario. To implement the bisection algorithm, we run the Python method for scalar root finding from the scipy library, in the scipy.optimize package, bisect, with a tolerance of roughly $10^{-13}$, yielding the first twelve correct digits of $\nu_* = 0.211231196923\ndots$, from which we obtain $M(a,b,p)< \CR$ by direct computation. The competitive ratio of the optimal algorithm solving undisclosed RO has therefore been shown to be less than $\CR$, since for Instance \ref{instance} it attains a gambler-to-prophet ratio lower than this value, for all $n$ large enough. Thus no algorithm can achieve a better competitive ratio and RO is $\CR$-hard. Recall that since for Instance \ref{instance} undisclosed RO and disclosed RO are equivalent, the hardness of $\CR$ applies to both models. The separation from OS follows from the algorithm designed in \cite{PeTa22}, which ensures a competitive ratio for OS of $0.7251 > \CR$. This makes our proof independent from \cite{BubChip23}, as it does not require their refinement on the above algorithm, which could also be invoked here. 
    \end{proof}
    As a sanity check, we can assess through simulations the sharpness of the mathematical bounds, used to estimate the gambler-to-prophet ratio of the optimal algorithm $T$ on Instance \ref{instance}. The limiting bound $M(a,b,p)=\rat\ndots$ is essentially tight, since simulating the dynamic program for Instance \ref{instance} produces a gambler-to-prophet ratio of approximately $0.72354$ with $n=10^4$, $0.72349$ with $n=10^5$ and $n=10^6$. For the interested reader, the complete list of correct decimals of $M(a,b,p)$ is provided in \Cref{code}, where we also derive formally that the error committed on $\nu^*$ and $q(\nu^*)$. The reference to the shared code for the dynamic program used in the sanity check is also provided in \Cref{code}.

    \section{Conclusions}
    In this work we obtained: a new state-of-the-art hardness for RO of $\CR$ proved through asymptotic analytic techniques; a first non-simulation-assisted analysis of the online optimal algorithm separating RO from OS. The impact of the techniques used in the analysis is the following.
	\begin{itemize}[noitemsep, topsep=0pt, parsep=0pt, partopsep=0pt]
		\item They can be applied to a larger class of similar instances. For example, $n$ iid random variables with more than one nontrivial value could be used: simulations suggest that adding one more nontrivial value is likely to yield better hardness, with the analysis still relatively feasible. 
		\item They provide a rigorous foundation for the computational search of hard instances for RO, with possible tightness. Once our analytic strategy is executed, the problem is reduced to optimizing with respect to the parameters of the instance, which could be either approached mathematically or numerically. If approached numerically, as the conclusion of \Cref{mainth} shows, it can be simple enough to be carried out via theoretical bounds on the error of approximation. In particular, we only rely on a root finding routine, which is much simpler than simulating an upper bound for the expected return of RO within some set of parameters, which is the approach of \cite{BubChip23}. The latter has intrinsic limitations: one cannot simulate directly the expected return for RO, and potential tightness is therefore lost.
		\item They are likely to offer new insight into other open separation problems, such as the separation of undisclosed RO and disclosed RO.
	\end{itemize}
     \section*{Acknowledgements.} This research was partially supported by the EPSRC grant EP/W005573/1, the ERC CoG 863818 (ForM-SMArt) grant, the ANID Chile grant ACT210005 and the French Agence Nationale de la Recherche (ANR) under reference ANR-21-CE40-0020 (CONVERGENCE project). We would like to thank Jos\'{e} Correa and Bruno Ziliotto for their precious advice, and Mona Mohammadi and Roodabeh Safavi for early conversations.    
\newpage
\bibliographystyle{abbrv}
\bibliography{SeparatingRandom}
\clearpage
\appendix
\renewcommand{\appendixpagename}{Supplementary materials}
	\appendixpage

	\section{Supplements to Section \ref{asymp}}\label{suppasymp}
	In this section we prove the asymptotic results concerning the acceptance times of the optimal algorithm.
	\begin{proof}[Proof of \Cref{accept}]
	\begin{enumerate}[label={}, style=unboxed, leftmargin=0cm]
	\item~
	\begin{enumerate}[a), style=unboxed, leftmargin=0cm]	
	\item Recall that $j_n$ is the smallest $i$ such that $a\geq \phi_i$, as per \Cref{defj}. Since for all $n$ large enough, $j_n\in[n]$, we consider $i\in [n]$, and in order to derive the sharp asymptotic estimate of the claim, we start by calculating explicitly $\phi_i$ by deriving $\phi_n,\ldots,\phi_{i+1}$ via backward induction. 
 \setlist[description]{font=\normalfont\itshape\space}
 \begin{description}
 \item[Step 1.] Note that since $a<b$, for all $k\geq j_n$, $\phi_k<b$. In fact by \Cref{backwardphi}, $\phi_{k-1}=\E(V\vee\phi_k)\geq \phi_k$ for all $k\in[n]$ and $\phi_{j_n}<a<b$. By \Cref{backwardphi} for $k=n$, as $n\longrightarrow\infty$,
	\begin{align}\label{base_j}
	    \phi_n = \frac{1 + bp}{n}\,,
	\end{align}
	and for all $i+1\leq k<n+1$ we will iterate \Cref{base_j} through \Cref{backwardphi,backward}, that is through
	\begin{align}\label{indformula1}
		\phi_{k-1}&=\E(V\vee\phi_k)=\frac{1}{n}+\frac{p(b\vee\phi_k)}{n}+\left(1-          \frac{p}{n}-\frac{1}{n^2}\right)\phi_k=\frac{1+pb}{n}+\left(1-\frac{p}{n}-\frac{1}{n^2}\right)\phi_k,
	\end{align} 
	which follows by expanding the expectation of the maximum as in \Cref{backward} and using $0<\phi_k\leq n$ as observed in \Cref{trivial}, along with the aforementioned fact, that for all $i+1\leq k< n+1$, $b\vee\phi_k=b$. One more iteration will suffice to clarify what the induction hypothesis should be. Consider $k=n$ in \Cref{indformula1}, then we have that
	   \begin{align*}
		\phi_{n-1}&=\frac{1+pb}{n}+\left(1-\frac{p}{n}-\frac{1}{n^2}\right)\frac{1+pb}{n}=\frac{1+pb}{n}\left[1+\left(1-\frac{p}{n}-\frac{1}{n^2}\right)\right].
	\end{align*}
	The induction hypothesis is therefore that for $i+1\leq k<n+1$ and $n$ large enough
		\begin{equation}\label{induction_j} 
			\phi_{k}=\frac{1+bp}{n}\sum_{j=0}^{n-k}\left(1-\frac{p}{n}-\frac{1}{n}\right)^j.	
		\end{equation}
	By assuming \Cref{induction_j} and using \Cref{indformula1}, it follows that
    \begin{align*}
        \phi_{k-1}&=\frac{1+pb}{n}+\left(1-\frac{p}{n}-\frac{1}{n^2}\right)\left[\frac{1+bp}{n}\sum_{j=0}^{n-k}\left(1-\frac{p}{n}-\frac{1}{n}\right)^j\right]=\frac{1+bp}{n}\sum_{j=0}^{n-k+1}\left(1-\frac{p}{n}-\frac{1}{n}\right)^j
    \end{align*}
    and the induction step is complete. 
    \item[Step 2.] Having shown \Cref{induction_j} for all $i\leq k\leq n$, we take $k=i$, so as to obtain
				\begin{equation}\label{exactgamma_i}
					\phi_i=\frac{1+bp}{n}\sum_{j=0}^{n-i}\left(1-\frac{p}{n}-\frac{1}{n}\right)^j= \frac{1+bp}{n}\left[\frac{1-\left(1-\frac{p}{n}-\frac{1}{n^2}\right)^{n-i+1}}{\frac{p}{n}+\frac{1}{n^2}}\right].
				\end{equation} 
			Thus  \begin{equation}\label{approxgamma_i}
					\phi_i=\frac{1+bp}{p}\left[1-\left(1-\frac{p}{n}\right)^{n-i+1}\right]+\mathcal{O}\left(\frac{1}{n}\right),
				\end{equation}
				having used
				\begin{equation}\label{1+bp/p}
					\frac{\frac{1+bp}{n}}{\frac{p}{n}+\frac{1}{n^2}}=\frac{1+bp}{p+\frac{1}{n}}=\frac{1+bp}{p}\left(1+\mathcal{O}\left(\frac{1}{n}\right)\right)
				\end{equation}
				and, upon factorising 
				\begin{equation}\label{factorise}
					1-\frac{p}{n}-\frac{1}{n^2}=\left(1-\frac{p}{n}\right)\left(1-\frac{1+\smallO(1)}{n^2}\right),
				\end{equation} 
				having used
				\begin{equation}\label{1-1/n^2}
					\left(1-\frac{1+\smallO(1)}{n^2}\right)^{n-i+1}=e^{-(n-i+1)\frac{1+\smallO(1)}{n^2}+\mathcal{O}\left(\frac{1}{n^3}\right)}=1+\mathcal{O}\left(\frac{1}{n}\right).
				\end{equation}
				Plugging \Cref{approxgamma_i} into \Cref{defj}, we observe that since $j_n$ is the smallest $i$ such that $a\geq \phi_i$, one obtains equivalently that $j_n$ is the smallest $i\in [n]$ such that
				\begin{equation}\label{def_j}
					a\geq \frac{1+bp}{p}\left[1-\left(1-\frac{p}{n}\right)^{n-i+1}\right]+\mathcal{O}\left(\frac{1}{n}\right).
				\end{equation}
			 \item[Step 3.] We rearrange \Cref{def_j} into
				\begin{equation}\label{def_j_log}
					i\geq n+1-\frac{\log\left(1-\frac{pa}{1+bp}+\mathcal{O}\left(\frac{1}{n}\right)\right)}{\log\left(1-\frac{p}{n}\right)}
				\end{equation}
				and then use the Taylor expansion of the logarithms as $n\longrightarrow\infty$, yielding
				\begin{align*}
					\frac{\log\left(1-\frac{pa}{1+bp}+\mathcal{O}\left(\frac{1}{n}\right)\right)}{\log\left(1-\frac{p}{n}\right)}&=\frac{\log\left(1-\frac{pa}{1+bp}+\mathcal{O}\left(\frac{1}{n}\right)\right)}{-\frac{p}{n}+\mathcal{O}\left(\frac{1}{n^2}\right)}=-\frac{n}{p}\frac{\log\left(\frac{1+(b-a)p}{1+bp}\right)+\mathcal{O}\left(\frac{1}{n}\right)}{1+\mathcal{O}\left(\frac{1}{n}\right)}\\&=-\frac{n}{p}\log\left(\frac{1+(b-a)p}{1+bp}\right)+\mathcal{O}(1).
				\end{align*}
				Plugging the expansion into \Cref{def_j_log} yields
				\[i\geq n+1+\frac{n}{p}\log\left(\frac{1+(b-a)p}{1+bp}\right)+\mathcal{O}(1),\] from which it follows that $j_n$ is the smallest $i\in [n]$ such that, as $n\longrightarrow\infty$,
				\[i\geq n\left[1+\frac{1}{p}\log\left(\frac{1+(b-a)p}{1+bp}\right)\right]+\mathcal{O}\left(1\right).\] Thus by the standard asymptotics of the ceiling function, as $n\longrightarrow\infty$, we have that
				\[j_n=\bigg\lceil n\left[1+\frac{1}{p}\log\left(\frac{1+(b-a)p}{1+bp}\right)\right]+\mathcal{O}\left(1\right)\bigg\rceil\sim n\left[1+\frac{1}{p}\log\left(\frac{1+(b-a)p}{1+bp}\right)\right].\]
			Note that since $0<a<1<b$, the coefficient of $n$ is positive by the inequality \[\log (1+x)\geq \frac{x}{1+x}\] for all $x>-1$, which implies, by taking \[x=-\frac{ap}{1+bp},\] that
   \[\log\left(\frac{1+(b-a)p}{1+bp}\right)=\log\left(1-\frac{ap}{1+bp}\right)>\frac{-\frac{ap}{1+bp}}{1-\frac{ap}{1+bp}}=-\frac{ap}{1+(b-a)p}>-p,\] yielding the claim.
    \end{description}				
        \item Replacing $a$ with $b$ in \Cref{def_j} shows that, by \Cref{defk}, $k_n$ is the smallest $2\leq i\leq n$ such that, as $n\longrightarrow\infty$,
	\[i\geq n\left[1+\frac{1}{p}\log\left(\frac{1}{1+bp}\right)\right]+\mathcal{O}\left(1\right),\] which yields \Cref{accept} (b) through a similar concluding argument as that of \Cref{accept} (a). Note that the coefficient of $n$ is positive by the condition $\log(1+pb)<p$ on the parameters of Instance \ref{instance}.
                    
	\item Noting that for Instance \ref{instance} we have $b>a>0$ and $p>0$, it follows that $1<1+(b-a)p$, and therefore \Cref{accept} (a, b) immediately imply that eventually $k_n\leq j_n$, since  \[1+\frac{1}{p}\log\frac{1}{1+bp}<1+\frac{1}{p}\log\frac{1+(b-a)p}{1+bp}.\]
				
	\item To show that $\bar{k}_n\geq k_n$ as $n\longrightarrow\infty$, assume by contradiction that there exists a subsequence $\{n_l\}$  such that $\bar{k}_{n_l}< k_{n_l}$. For simplicity, we relabel the indices with $n$. Thus our hypothesis is that for infinitely many $n$, $\bar{k}_{n}< k_{n}$. By \Cref{defbark}, $\bar{k}_{n}$ is the smallest $k$ such that $\bar{\phi}_k\leq b$, so it follows that for all $n$ considered, 
    \begin{equation}\label{contradiction}\bar{\phi}_{k_n-1}\leq b.\end{equation} 
    Our strategy will be to derive a contradiction with \Cref{contradiction} thanks to an iterative lower bound on $\bar{\phi}_{k_n-1}$. 
    \setlist[description]{font=\normalfont\itshape\space}
    \begin{description}
    \item[Step1.] First of all note that by \Cref{mono}, also for all $k\geq k_n$, $\bar{\phi}_{k}\leq b$. In this first step we exploit this fact to derive a lower bound on $\bar{\phi}_k$ for all $k\geq k_n-1$. By \Cref{backwardphi,backward} and the facts aforementioned, we have that
\begin{align}\label{backwardlower}
\bar{\phi}_k&\geq\frac{a}{n-k+1}+\left(1-\frac{1}{n-k+1}\right)\left(\frac{1}{n}+p\frac{b\vee\bar{\phi}_{k+1}}{n}+\bar{\phi}_{k+1}\left(1-\frac{p}{n}-\frac{1}{n^2}\right)\right)\notag\\&=\frac{a}{n-k+1}+\left(1-\frac{1}{n-k+1}\right)\left(\frac{1+bp}{n}+\bar{\phi}_{k+1}\left(1-\frac{p}{n}-\frac{1}{n^2}\right)\right)
\end{align}
Then by \Cref{backwardlower} applied with $k=n-1$, using $\bar{\phi}_{n}=a$, we derive
\begin{equation}\label{base_k}
\bar{\phi}_{n-1}\geq a+\frac{1}{2}\left[\frac{1+(b-a)p}{n}-\frac{a}{n^2}\right].
\end{equation}
Iterating from \Cref{base_k} via \Cref{backwardlower} we obtain, by induction, that for all $k\geq k_n-1$,
\begin{equation}\label{indk}
\bar{\phi}_k\ge a+\left(\frac{1+(b-a)p}{n}-\frac{a}{n^2}\right)\sum_{j=0}^{n-k-1}\frac{n-k-j}{n+1-k}\left(1-\frac{p}{n}-\frac{1}{n^2}\right)^j.
\end{equation}
In fact if \Cref{indk} is true for any $k_n<k<n+1$, by \Cref{backwardlower} applied to $k-1$ we obtain
\begin{align*}
\bar{\phi}_{k-1}&\ge \frac{a}{n-k+2}+\left(1-\frac{1}{n-k+2}\right)\bigg[\frac{1+pb}{n}+a\left(1-\frac{p}{n}-\frac{1}{n^2}\right)\\&+\left(\frac{1+(b-a)p}{n}-\frac{a}{n^2}\right)\sum_{j=0}^{n-k-1}\frac{n-k-j}{n+1-k}\left(1-\frac{p}{n}-\frac{1}{n^2}\right)^{j+1}\bigg]\\& =a+\left(\frac{1+(b-a)p}{n}-\frac{a}{n^2}\right)\frac{n-k+1}{n-k+2}\left[1+\sum_{j=1}^{n-k}\frac{n-k-(j-1)}{n-(k-1)}\left(1-\frac{p}{n}-\frac{1}{n^2}\right)^j\right]\\&=a+\left(\frac{1+(b-a)p}{n}-\frac{a}{n^2}\right)\sum_{j=0}^{n-k}\frac{n-(k-1)-j}{n+1-(k-1)}\left(1-\frac{p}{n}-\frac{1}{n^2}\right)^j.
\end{align*}
\item[Step 2.] In this step we derive a sharp asymptotic estimate for \Cref{indk}. Let $q= q(n,p)\defeq 1-\sfrac{p}{n}-\sfrac{1}{n^2}$. The summation term can be rewritten as
\begin{align*}
&S_{k,n}\defeq\sum_{j=0}^{n-k-1}\frac{n+1-k-(j+1)}{n+1-k}q^j=\sum_{j=0}^{n-k-1}q^j-\frac{1}{n+1-k}\sum_{j=0}^{n-k-1}(j+1)q^j=\\&\frac{1-q^{n-k}}{1-q}-\frac{1}{n+1-k}\frac{d}{dq}\left(\sum_{j=0}^{n-k}q^j\right)=\frac{1-q^{n-k}}{1-q}-\frac{1}{n+1-k}\frac{d}{dq}\left(\frac{1-q^{n-k+1}}{1-q}\right)=\\&\frac{1-q^{n-k}}{1-q}-\frac{1}{n+1-k}\frac{1-q^{n-k}[(n-k)(1-q)+1]}{(1-q)^2},
\end{align*}
so by factorising $q$ as in \Cref{factorise}, and exploiting
\begin{equation}\label{e^-p}
\left(1-\frac{p}{n}\right)^n=e^{-p}+\mathcal{O}\left(\frac{1}{n}\right),
\end{equation}
we have that for all $k\geq k_n-1$,
\begin{align*}
S_{k,n}&=\frac{1-\left(1-\frac{p}{n}-\frac{1}{n^2}\right)^{n-k}}{\frac{p}{n}+\frac{1}{n^2}}-\frac{1}{n+1-k}\frac{1-\left(1-\frac{p}{n}-\frac{1}{n^2}\right)^{n-k}\left[(n-k)\left(\frac{p}{n}+\frac{1}{n^2}\right)+1\right]}{\left(\frac{p}{n}+\frac{1}{n^2}\right)^2}\\&=\frac{n}{p}\left[1-e^{-p\left(1-\frac{k}{n}\right)}+\mathcal{O}\left(\frac{1}{n}\right)\right]-\frac{n}{p^2\left(1-\frac{k}{n}\right)}\left\lbrace1-e^{-p\left(1-\frac{k}{n}\right)}\left[1+p\left(1-\frac{k}{n}\right)\right]+\mathcal{O}\left(\frac{1}{n}\right)\right\rbrace\\&=\frac{n}{p}\left[1-\frac{1-e^{-p\left(1-\frac{k}{n}\right)}}{p\left(1-\frac{k}{n}\right)}+\mathcal{O}\left(\frac{1}{n\left(1-\frac{k}{n}\right)}\right)\right].
\end{align*}
Recall that by \Cref{accept} (b), \[\frac{k_n}{n}\longrightarrow 1-\frac{\log(1+bp)}{p}.\] For $k = k_n-1$, the error term
\[\mathcal{O}\left(\frac{1}{n\left(1-\frac{k}{n}\right)}\right)=\mathcal{O}\left(\frac{1}{n}\right).\]
\item[Step 3.] We plug the asymptotic estimate of $S_{k_n-1,n}$ into \Cref{indk}, for $k=k_n-1$, and we obtain
\begin{align*}
\bar{\phi}_{k_n-1}&\geq a+\left(\frac{1+(b-a)p}{p}-\frac{a}{pn}\right)\left[1-\frac{1-e^{-p\left(1-\frac{k_n-1}{n}\right)}}{p\left(1-\frac{k_n-1}{n}\right)}+\mathcal{O}\left(\frac{1}{n}\right)\right]\\&=a+\frac{1+(b-a)p}{p}-\frac{1+(b-a)p}{p}\frac{1-e^{-p\left(1-\frac{k_n}{n}\right)}}{p\left(1-\frac{k_n}{n}\right)}+\mathcal{O}\left(\frac{1}{n}\right)\\&=b+\frac{1}{p}-\frac{1+(b-a)p}{p}\frac{1-e^{-p\left(1-\frac{k_n}{n}\right)}}{p\left(1-\frac{k_n}{n}\right)}+\mathcal{O}\left(\frac{1}{n}\right).
\end{align*}
Applying \Cref{accept} (b) yields
\begin{align*}
\bar{\phi}_{k_n-1}&\geq b+\frac{1}{p}-\frac{1+(b-a)p}{p}\frac{1-e^{-p\left(1-\frac{k_n}{n}\right)}}{p\left(1-\frac{k_n}{n}\right)}+\mathcal{O}\left(\frac{1}{n}\right)\\&=b+\frac{1}{p}-\frac{1+(b-a)p}{p}\frac{1-\frac{1}{1+bp}}{\log(1+pb)}+\smallO(1)=b+\frac{1}{p}-\frac{b[1+(b-a)p]}{(1+pb)\log(1+pb)}+\smallO(1).
\end{align*}
By Condition \ref{conditionV}, which ensures that \[\frac{1}{p}-\frac{b[1+(b-a)p]}{(1+pb)\log(1+pb)}>0,\] and $n$ being arbitrarily large, we have that $\bar{\phi}_{k_n-1}>b$, which contradicts \Cref{contradiction}. The assumption that there are infinitely many $n$ such that $\bar{k}_n<k_n$ is therefore false, meaning that for all $n$ large enough, $k_n\leq\bar{k}_n$.
\end{description}
\item To show that $\bar{k}_n\leq j_n$ as $n\longrightarrow\infty$, assume by contradiction that there exists a subsequence $\{n_l\}$ such that $\bar{k}_{n_l}> j_{n_l}$ as $l\longrightarrow\infty$. For simplicity, relabel $n_l$ as $n$, thus starting the argument, without loss of generality, with the assumption by contradiction, that $\bar{k}_n> j_n$ for infinitely many $n$. The overall strategy will be the following: under the assumption by contradiction we derive an upper bound on $\bar{k}_n$, which we recall to be the smallest integer $k$ such that $b\geq\bar{\phi}_{k}$ as per \Cref{defbark}. Thanks to this upper bound we will show that $\sfrac{\bar{k}_n}{j_n}<1$, contradicting the assumption that $\sfrac{\bar{k}_n}{j_n}>1$. To obtain such an upper bound on $\bar{k}_n$ we will find first a suitable upper bound on $\bar{\phi}_{\bar{k}_n}$.
        \begin{description}[style=unboxed]
        \item[Step 1.] Since for all $\bar{k}_n\leq k\leq n$, we have that $k\geq j_n$, not only we know that $\bar{\phi}_k\leq b$, but we also know that $\phi_k\leq a$, and therefore by \Cref{backwardphi,backward} we obtain that
        \begin{equation}\label{backwardupper}
            \bar{\phi}_k\leq\frac{a}{n-k+1}+\left(1-\frac{1}{n-k+1}\right)\left[\frac{1+bp}{n}+\bar{\phi}_{k+1}\left(1-\frac{p}{n}\right)\right].
        \end{equation}
        Recall that $\bar{\phi}_{n}=a$. Then by \Cref{backwardupper} it follows that
        \begin{equation}\label{baseupper}
            \bar{\phi}_{n-1}\leq\frac{a}{2}+\frac{1}{2}\left[\frac{1+bp}{n}+a\left(1-\frac{p}{n}\right)\right]=a+\frac{1+(b-a)p}{2n}.
        \end{equation}
        The induction hypothesis will be that for any $\bar{k}_n\leq k\leq n-1$,
        \begin{equation}\label{inductupper}
            \bar{\phi}_{k+1}\leq a+\frac{n-k-1}{2}\frac{1+(b-a)p}{n}.
        \end{equation}
        Then by \Cref{inductupper,backwardupper} it follows that 
        \begin{align*}
            \bar{\phi}_k&\leq\frac{a}{n-k+1}+\left(1-\frac{1}{n-k+1}\right)\left[\frac{1+bp}{n}+\left(a+\frac{n-k-1}{2}\frac{1+(b-a)p}{n}\right)\left(1-\frac{p}{n}\right)\right]\\&=a+\left(1-\frac{1}{n-k+1}\right)\left[\frac{n-k+1}{2}\frac{1+(b-a)p}{n}-\frac{p}{n}\frac{n-k-1}{2}\frac{1+(b-a)p}{n}\right]\\&\leq a+\left(1-\frac{1}{n-k+1}\right)\frac{n-k+1}{2}\frac{1+(b-a)p}{n}=a+\frac{n-k}{2}\frac{1+(b-a)p}{n}.
        \end{align*}
        By induction on $k$ we obtain that 
        \begin{equation}\label{middlestage}
            \bar{\phi}_{\bar{k}_n}\leq a+\frac{n-\bar{k}_n}{2}\frac{1+(b-a)p}{n}.
        \end{equation}
        By a trivial induction argument one can iterate this bound until time $j_n$; showing the first step will suffice. Since for all $j_n\leq k<\bar{k}_n$, we have that $b\vee\bar{\phi}_k=\bar{\phi_k}$ and $a\vee\phi_k=a$, by \Cref{backward,backwardphi} it follows that
        \begin{align}\label{backwardupper_}
            \bar{\phi}_k&\leq\frac{a}{n-k+1}+\left(1-\frac{1}{n-k+1}\right)\left[\frac{1}{n}+\frac{p}{n}\bar{\phi}_{k+1}+\bar{\phi}_{k+1}\left(1-\frac{p}{n}\right)\right]\notag\\&=\frac{a}{n-k+1}+\left(1-\frac{1}{n-k+1}\right)\left[\frac{1}{n}+\bar{\phi}_{k+1}\right].
        \end{align}
        In the induction steps past time $\bar{k}_n$, \Cref{backwardupper_} will take the place of \Cref{backwardupper}. We show the first step. By \Cref{backwardupper_,middlestage} and the fact that \[\frac{1}{n}<\frac{1+(b-a)p}{n},\] we have that
        \begin{align*}
            \bar{\phi}_{\bar{k}_n-1}&\leq \frac{a}{n-\bar{k}_n+2}+\left(1-\frac{1}{n-\bar{k}_n+2}\right)\left[\frac{1}{n}+a+\frac{n-\bar{k}_n}{2}\frac{1+(b-a)p}{n}\right]\\&\leq a+\left(1-\frac{1}{n-\bar{k}_n+2}\right)\frac{n-\bar{k}_n+2}{2}\frac{1+(b-a)p}{n}= a+\frac{n-\bar{k}_n+1}{2}\frac{1+(b-a)p}{n},
        \end{align*}
       The mechanism of this iteration is trivial, due to the cancellation of the fractions carrying over for all $k$, and therefore what we obtained for the previous steps can be iterated by induction for all successive steps. Thus we have shown that for all $j_n\leq k\leq n-1$,
        \begin{equation}\label{claimk}
            \bar{\phi}_{k}\leq a+\frac{n-k}{2}\frac{1+(b-a)p}{n}.
        \end{equation}
        \item[Step 2.] Consider 
        \[
            k_n^*\defeq\inf\left\lbrace k\geq j_n:\:b\geq a+\frac{n-k}{2}\frac{1+(b-a)p}{n}\right\rbrace \,.
        \] 
        Equivalently, $k_n^*$ is the smallest $k\geq j_n$ such that 
        \[
            k \geq n\frac{1-(2-p)(b-a)}{1+p(b-a)} \,.
        \] 
        Note that $k_n^* \le n$ due to $b > a$, and $k_n^*  > 0$ due to Condition \ref{conditionII}. 
        Then, by a similar reasoning  as in the conclusion of \Cref{accept} (a), 
        \[
            k_n^*=\left\lceil n\frac{1-(2-p)(b-a)}{1+p(b-a)} \right\rceil\sim n\frac{1-(2-p)(b-a)}{1+p(b-a)} \,.
        \]
        Note that 
        \begin{equation}
            \label{barstar}
            \bar{k}_n\leq k_n^* \,,
        \end{equation} 
        since, by \Cref{claimk}, the earliest $k\geq j_n$ such that $b\geq\bar{\phi}_k$ is smaller than the earliest $k\geq j_n$ such that 
        \[
            b\geq a+\frac{n-k}{2}\frac{1+(b-a)p}{n} \,.
        \]
        
        \item[Step 3.] Recall that we are assuming that $j_n <  \bar{k}_n$. Equivalently, we have that 
        \[
            1+\frac{1}{j_n}\leq\frac{\bar{k}_n}{j_n} \,.
        \] 
        By \Cref{barstar}, Condition \ref{conditionIII}, and \Cref{accept} (a), we reach the following contradiction: as $n\longrightarrow\infty$,
        \[ \frac{\bar{k}_n}{j_n}\geq 1+\frac{1}{j_n}\longrightarrow 1+\frac{1}{1+\frac{1}{p}\log\frac{1+p(b-a)}{1+pb}}>1\]
               and 
                \[\frac{\bar{k}_n}{j_n}\leq\frac{k_n^*}{j_n}
                \longrightarrow\frac{\frac{1-(2-p)(b-a)}{1+p(b-a)}}{1+\frac{1}{p}\log\frac{1+p(b-a)}{1+pb}}
                < 1 \,.
        \] 
        Thus it must hold that for all $n$ large enough, $\bar{k}_n\leq j_n$.
        \end{description}
	\end{enumerate}
        \end{enumerate}
	\end{proof}
	\section{Supplements to Section \ref{hardness}}\label{suppmax}
	In this section we start by computing the exact asymptotics for the expectation of the maximum for Instance \ref{instance}.
	\begin{proof}[Proof of \Cref{max}]
		Since by assumption $b>a$, for all $n$ large enough
		\begin{equation*}
		      \max_{i\in[n+1]}V_i\sim
			\begin{cases}
				n,&\text{w.p. }\:1-\left(1-\frac{1}{n^2}\right)^n,\\
				b,&\text{w.p. }\:\left(1-\frac{1}{n^2}\right)^n-\left(1-\frac{p}{n}-\frac{1}{n^2}\right)^n,\\
				a,&\text{w.p. }\:\left(1-\frac{p}{n}-\frac{1}{n^2}\right)^n.
			\end{cases}		
		\end{equation*}
		Thus \[\E\max_{i\in[n+1]}V_i=n\left[1-\left(1-\frac{1}{n^2}\right)^n\right]+b\left[\left(1-\frac{1}{n^2}\right)^n-\left(1-\frac{p}{n}-\frac{1}{n^2}\right)^n\right]+a\left(1-\frac{p}{n}-\frac{1}{n^2}\right)^n.\]
		By a similar estimate as in \Cref{1-1/n^2}, factorising as in \Cref{factorise}, we have that
		\[\left(1-\frac{1}{n^2}\right)^n=1-\frac{1}{n}+\mathcal{O}\left(\frac{1}{n^3}\right),\] so that by exploiting \Cref{e^-p} it follows that
		\[\E\max_{i\in[n+1]}V_i=n\left[\frac{1}{n}+\mathcal{O}\left(\frac{1}{n^3}\right)\right]+b\left[1-\frac{1}{n}+\mathcal{O}\left(\frac{1}{n^3}\right)-e^{-p}+\mathcal{O}\left(\frac{1}{n}\right)\right]+a\left[e^{-p}+\mathcal{O}\left(\frac{1}{n}\right)\right]\]
		and the claim follows.
	\end{proof}
	Next, we compute a sharp asymptotic estimate for the expectation of the optimal algorithm $\E V_{\pi_T}$. We will do so through the law of total expectation with respect to the partitioning $\{\Omega_1,\ldots,\Omega_{n+1}\}$ of the sample space $\Omega$. The role of \Cref{accept} in this computation can be appreciated from the following. Consider $i<j_n$ for instance: conditionally on $\Omega_i$, the optimal stopping rule does not stop when $a$ or $0$ are probed (except for $0$, in the last step), but it stops the first time $n$ is probed, or, at certain times, when $b$ is probed. Thus to determine when $b$ is accepted or not, knowledge of the relative position of $\bar{k}_n$ with respect to $k_n$ and $j_n$ is needed. Thanks to \Cref{accept} we avoid a lengthy case analysis, since only one ordering is possible: $k_n\leq \bar{k}_n\leq j_n$. 
	\begin{proof}[Proof of \Cref{optimal}]
		We start by finding the distribution of $(V_{\pi_T}|\Omega_i)$ for all $i\in[n+1]$ (the abuse of notation denotes the conditional distribution of $V_{\pi_T}$ given $\Omega_i$), so as to determine, for all $n$ large enough, $\E_iV_{\pi_T}$ for all $i\in[n+1]$. Recall that, conditionally on $\Omega_i$, at every step of the process other than the $i$th, independently, the algorithm could probe $n$, $b$ and $0$, with probabilities $\sfrac{1}{n^2}$, $\sfrac{p}{n}$ and $1-\sfrac{p}{n}-\sfrac{1}{n^2}$ respectively. While $n$ is always accepted, $0$ is never accepted, until the end, by the optimal algorithm.
	\begin{enumerate}[i)]
				\item If $i< k_n$, conditionally on $\Omega_i$, $a$ is probed before time $k_n\leq j_n$ by \Cref{accept} (c), so $a$ will be rejected. There are in total $n$ steps, in which $n$ could be probed and accepted. Before time $k_n$, the optimal algorithm stops with reward $n$ if and only if any of the first $k_n-2$ values probed is $n$, because if $b$ were probed before time $i$, the algorithm would not stop by definition of $\bar{k}_n$, since $i<k_n\leq \bar{k}_n$ by \Cref{accept} (d); if $b$ were probed between time $i$ and $k_n$, the algorithm would not stop by definition of $k_n$. Thus the probability of stopping with reward $n$ before time $k_n$ is \[1-\left(1-\frac{1}{n^2}\right)^{k_n-2}.\] From time $k_n$ onward, at each step, the optimal algorithm stops with reward $n$ if and only if all previous values probed, starting from time $k_n$, are $0$. In fact if after or at time $k_n$, $b$ were to be probed before $n$, the algorithm would stop with reward $b$ by definition of $k_n>i$. Thus the probability of stopping with reward $n$ at each step from time $k_n$ onward (there are $n+2-k_n$ such steps), is the probability of not probing $n$ in any of the first $k_n-2$ steps (this guarantees that the algorithm reaches the $k_n$th step), multiplied by the probability of always probing $0$ from time $k_n$ up to the first time $n$ is probed: 
                \[
                    \frac{1}{n^2}\left(1-\frac{1}{n^2}\right)^{k_n-2}\:\sum_{j=0}^{n+1-k_n}\left(1-\frac{p}{n}-\frac{1}{n^2}\right)^j \,.
                \]	
                There are in total $n$ steps, in which $b$ could be probed, but only in $n+2-k_n$ of these, from time $k_n$ onward, $b$ would be accepted. From the previous comments on the acceptance of $b$, at each of these steps, the probability that the optimal algorithm stops with reward $b$ is the probability of not probing $n$ in any of the first $k_n-2$ steps (this guarantees that the algorithm reaches the $k_n$th step), multiplied by the probability that all remaining probed values are $0$, from time $k_n$ up to the first time $b$ is probed. Thus the probability of stopping with reward $b$ is \[\frac{p}{n}\left(1-\frac{1}{n^2}\right)^{k_n-2}\:\sum_{j=0}^{n+1-k_n}\left(1-\frac{p}{n}-\frac{1}{n^2}\right)^j.\] In conclusion, for every $i<k_n$,
				\begin{equation}\label{ii:i<ka}	(V_{\pi_T}|\Omega_i)\sim\begin{cases}n,&\text{w.p. }1-\left(1-\frac{1}{n^2}\right)^{k_n-2}+\frac{1}{n^2}\left(1-\frac{1}{n^2}\right)^{k_n-2}\sum_{j=0}^{n+1-k_n}\left(1-\frac{p}{n}-\frac{1}{n^2}\right)^j\\b,&\text{w.p. }\frac{p}{n}\left(1-\frac{1}{n^2}\right)^{k_n-2}\sum_{j=0}^{n+1-k_n}\left(1-\frac{p}{n}-\frac{1}{n^2}\right)^j \\0,&\text{otherwise.}\end{cases}
				\end{equation}
				Since
				\[\sum_{j=0}^{n+1-k_n}\left(1-\frac{p}{n}-\frac{1}{n^2}\right)^j =\frac{1-\left(1-\frac{p}{n}-\frac{1}{n^2}\right)^{n+2-k_n} }{\frac{p}{n}+\frac{1}{n^2}}\]
				and from \Cref{1-1/n^2,e^-p} it follows that
				\[\left(1-\frac{p}{n}-\frac{1}{n^2}\right)^{k_n-2}=\left(1-\frac{p}{n}\right)^{k_n}+\mathcal{O}\left(\frac{1}{n}\right)=\left[\left(1-\frac{p}{n}\right)^n\right]^{\frac{k_n}{n}}+\mathcal{O}\left(\frac{1}{n}\right)=e^{-p\frac{k_n}{n}}+\mathcal{O}\left(\frac{1}{n}\right)\]
				and
				\[\left(1-\frac{p}{n}-\frac{1}{n^2}\right)^{n+2-k_n}=\frac{e^{p\frac{k_n}{n}}}{e^p}+\mathcal{O}\left(\frac{1}{n}\right)=e^{p\left(\frac{k_n}{n}-1\right)}+\mathcal{O}\left(\frac{1}{n}\right),\]
				we conclude that 
				\begin{equation}\label{sum}
					\sum_{j=0}^{n+1-k_n}\left(1-\frac{p}{n}-\frac{1}{n^2}\right)^j =\frac{1-e^{p\left(\frac{k_n}{n}-1\right)}+\mathcal{O}\left(\frac{1}{n}\right)}{\frac{p}{n}+\frac{1}{n^2}}=\frac{n}{p}\left[1-e^{p\left(\frac{k_n}{n}-1\right)}+\mathcal{O}\left(\frac{1}{n}\right)\right].
				\end{equation} 
				Thus
				\begin{equation}\label{sumprod1}
					\frac{p}{n}\sum_{j=0}^{n+1-k_n}\left(1-\frac{p}{n}-\frac{1}{n^2}\right)^j =1-e^{p\left(\frac{k_n}{n}-1\right)}+\mathcal{O}\left(\frac{1}{n}\right)
				\end{equation}
				and
				\begin{equation}\label{sumprod2}
					\frac{1}{n^2}\sum_{j=0}^{n+1-k_n}\left(1-\frac{p}{n}-\frac{1}{n^2}\right)^j =\frac{1}{pn}\left(1-e^{p\left(\frac{k_n}{n}-1\right)}\right)+\mathcal{O}\left(\frac{1}{n^2}\right).
				\end{equation}
				Considering that $\sfrac{k_n}{n}$ is subunitary and bounded away from $0$ and $1$ by \Cref{accept} (b), it also follows that 
				\begin{align}\label{1-1/n^2kn}
					\left(1-\frac{1}{n^2}\right)^{k_n-2}&=\left(1-\frac{1}{n^2}\right)^{k_n}+\mathcal{O}\left(\frac{1}{n^2}\right)=e^{-\frac{k_n}{n^2}+\mathcal{O}\left(\frac{k_n}{n^4}\right)}+\mathcal{O}\left(\frac{1}{n^2}\right)\notag\\&=1-\frac{k_n}{n^2}+\mathcal{O}\left(\frac{k_n}{n^4}\right)+\mathcal{O}\left[\left(\frac{k_n}{n^2}+\mathcal{O}\left(\frac{k_n}{n^4}\right)\right)^2\right]+\mathcal{O}\left(\frac{1}{n^2}\right)\notag\\&=1-\frac{k_n}{n^2}+\mathcal{O}\left(\frac{1}{n^2}\right).
				\end{align} 
				In conclusion plugging \Cref{sumprod1,sumprod2,1-1/n^2kn} into \Cref{ii:i<ka} yields
				\begin{align}\label{ii:Omega_i<ka}
					\E_iV_{\pi_T}&=\notag\\&n\left[\frac{k_n}{n^2}+\mathcal{O}\left(\frac{1}{n^2}\right)\right]+\left[1-\frac{k_n}{n^2}+\mathcal{O}\left(\frac{1}{n^2}\right)\right]\left[\left(\frac{1}{p}+b\right)\left(1-e^{p\left(\frac{k_n}{n}-1\right)}\right)+\mathcal{O}\left(\frac{1}{n}\right)\right]\notag\\&=\frac{k_n}{n}+\left(\frac{1}{p}+b\right)\left(1-e^{p\left(\frac{k_n}{n}-1\right)}\right)+\mathcal{O}\left(\frac{1}{n}\right).
				\end{align}
    
				\item If $k_n \leq i < \bar{k}_n$ (assuming that there is any such $i$; if not, the empty sum convention used in \Cref{ii:lawtot_ka<kc<j} will take care of the case $k_n=\bar{k}_n$) conditionally on $\Omega_i$, $a$ is probed after or at time $k_n$, strictly before time $\bar{k}_n\leq j_n$ by \Cref{accept} (d, e), so $a$ will be rejected. There are in total $n$ steps, in which $n$ could be probed and accepted. Before time $i$, the optimal algorithm stops with reward $n$ if and only if any of the first $i-1$ values probed is $n$, because if $b$ were probed before time $i$, the optimal algorithm would not stop by definition of $\bar{k}_n>i$. Thus the probability of stopping with reward $n$ before time $i$ is \[1-\left(1-\frac{1}{n^2}\right)^{i-1}.\] 
                From time $i$ onward, at each step, the optimal algorithm stops with reward $n$ if and only if all the previous values probed, starting from time $i+1$, are $0$. In fact if after time $i\geq k_n$, $b$ were to be probed before $n$, the algorithm would stop with reward $b$ by definition of $k_n$. Thus the probability of stopping with reward $n$ at each step from time $i+1$ onward (there are $n+1-i$ such steps), is the probability of not probing $n$ in any of the first $i-1$ steps (which ensures that the algorithm reaches the $i$th step), multiplied by the probability of always probing $0$, from time $i+1$ up to the first time $n$ is probed: 
                \[
                    \frac{1}{n^2}\left(1-\frac{1}{n^2}\right)^{i-1}\:\sum_{j=0}^{n-i}\left(1-\frac{p}{n}-\frac{1}{n^2}\right)^j \,.
                \]
                There are in total $n$ steps, in which $b$ could be probed, but only in $n+1-i$ of these steps, from time $i+1$ onward, $b$ would be accepted, since $k_n\leq i<\bar{k}_n$. From the previous comments on the acceptance of $b$, at each of these steps, the probability that the optimal algorithm stops with reward $b$ is the probability of not probing $n$ in any of the first $i-1$ steps (which ensures that the algorithm reaches the $i$th step), multiplied by the probability that all remaining probed values are $0$, from time $i+1$ up to the first time $b$ is probed. Thus the probability of stopping with reward $b$ is 
                \[
                    \frac{p}{n}\left(1-\frac{1}{n^2}\right)^{i-1}\:\sum_{j=0}^{n-i}\left(1-\frac{p}{n}-\frac{1}{n^2}\right)^j \,.
                \] 
                In conclusion for every $k_n\leq i<\bar{k}_n$,
				\begin{equation}\label{ii:ka<i<kc}
					(V_{\pi_T}|\Omega_i)\sim\begin{cases}n,&\text{w.p. }1-\left(1-\frac{1}{n^2}\right)^{i-1}+\frac{1}{n^2}\left(1-\frac{1}{n^2}\right)^{i-1}\sum_{j=0}^{n-i}\left(1-\frac{p}{n}-\frac{1}{n^2}\right)^j\\b,&\text{w.p. }\frac{p}{n}\left(1-\frac{1}{n^2}\right)^{i-1}\sum_{j=0}^{n-i}\left(1-\frac{p}{n}-\frac{1}{n^2}\right)^j \\0,&\text{otherwise.}\end{cases}
				\end{equation}
				which is an expression similar to that of \Cref{ii:i<ka}, yielding, through analogous methods (which apply since $k_n\leq i<j_n$ by \Cref{accept} (a, b)),
				\begin{align}\label{ii:Omega_ka<i<kc}
					\E_iV_{\pi_T}=\frac{i}{n}+\left(\frac{1}{p}+b\right)\left(1-e^{p\left(\frac{i}{n}-1\right)}\right)+\mathcal{O}\left(\frac{1}{n}\right).
				\end{align}
				
				\item If $\bar{k}_n\leq i<j_n$ (assuming that there is any such $i$; if not, the empty sum convention used in \Cref{ii:lawtot_ka<kc<j} will take care of the case $\bar{k}_n=j_n$) conditionally on $\Omega_i$, $a$ is probed after or at time $\bar{k}_n$, before time $j_n$, so $a$ will be rejected. There are in total $n$ steps, in which $n$ could be probed and accepted. Before time $\bar{k}_n$, the optimal algorithm stops with reward $n$ if and only if any of the first $\bar{k}_n-1$ values probed is $n$, because if $b$ were probed before step $\bar{k}_n\leq i$, the algorithm would not stop by definition of $\bar{k}_n$. Thus the probability of stopping with reward $n$ before time $\bar{k}_n$ is \[1-\left(1-\frac{1}{n^2}\right)^{\bar{k}_n-1}.\] From time $\bar{k}_n$ onward, the optimal algorithm stops with reward $n$ if and only if all previous values probed (except for the $i$th), starting from time $k_n$, are $0$. In fact as of time $\bar{k}_n$, if $b$ were to be probed before $n$, the optimal algorithm would always stop with reward $b$: if probed between time $\bar{k}_n$ (included) and time $i$, by definition of $\bar{k}_n$; from time $i+1>\bar{k}_n\geq k_n$, by definition of $k_n$ and \Cref{accept} (d). Thus the probability of stopping with reward $n$ at each step from time $\bar{k}_n$ (there are $n+1-\bar{k}_n$ such steps, due to the exception of the $i$th step), is the probability of not probing $n$ in any of the first $\bar{k}_n-1$ steps (which ensures that the algorithm reaches the $\bar{k}_n$th step) multiplied by the probability of always probing $0$, from time $\bar{k}_n$ up to the first time $n$ is probed (skipping time $i$ when appropriate: the actual value of $i$ does not affect the final expression): 
				\[\frac{1}{n^2}\left(1-\frac{1}{n^2}\right)^{\bar{k}_n-1}\:\sum_{j=0}^{n-\bar{k}_n}\left(1-\frac{p}{n}-\frac{1}{n^2}\right)^j.\]
				There are in total $n$ steps, in which $b$ could be probed, but only in $n+1-\bar{k}_n$ of these, from time $\bar{k}_n$ onward, $b$ would be accepted. From the previous comments on the acceptance of $b$, at each of these steps, either between time $\bar{k}_n$ and $i$ or from time $i+1>\bar{k}_n\geq k_n$, the probability of the optimal algorithm stopping with reward $b$ is the probability of not probing $n$ in any of the first $\bar{k}_n-1$ steps (which ensures that the algorithm reaches the $\bar{k}_n$th step) multiplied by the probability that all remaining probed values are $0$, from time $\bar{k}_n$ up to the first time $b$ is probed (skipping time $i$ when appropriate: the actual value of $i$ does not affect the final expression): 
				\[\frac{p}{n}\left(1-\frac{1}{n^2}\right)^{\bar{k}_n-1}\:\sum_{j=0}^{n-\bar{k}_n}\left(1-\frac{p}{n}-\frac{1}{n^2}\right)^j.\]
				In conclusion for every $\bar{k}_n\leq i<j_n$,
				\begin{equation}\label{ii:kc<i<j}
                (V_{\pi_T}|\Omega_i)\sim\begin{cases}n,&\text{w.p. } 1-\left(1-\frac{1}{n^2}\right)^{\bar{k}_n-1}+\frac{1}{n^2}\left(1-\frac{1}{n^2}\right)^{\bar{k}_n-1}\sum_{j=0}^{n-\bar{k}_n}\left(1-\frac{p}{n}-\frac{1}{n^2}\right)^j\\b,&\text{w.p. }\frac{p}{n}\left(1-\frac{1}{n^2}\right)^{\bar{k}_n-1}\sum_{j=0}^{n-\bar{k}_n}\left(1-\frac{p}{n}-\frac{1}{n^2}\right)^j\\0,&\text{otherwise.}\end{cases}
				\end{equation}
				This is the same as \Cref{ii:i<ka}, except for formally having $\bar{k}_n+1$ instead of $k_n$ in the expression. This difference only contributes with $\mathcal{O}(\sfrac{1}{n})$-terms, so we can conclude, by similar methods, that
				\begin{equation}\label{ii:Omega_kc<i<j}
					\E_iV_{\pi_T}=\frac{\bar{k}_n}{n}+\left(\frac{1}{p}+b\right)\left(1-e^{p\left(\frac{\bar{k}_n}{n}-1\right)}\right)+\mathcal{O}\left(\frac{1}{n}\right).
				\end{equation}
    
				\item If $i\geq j_n$, conditionally on $\Omega_i$, $a$ is probed after or at time $j_n$, so $a$ will be accepted and the algorithm will not get past step $i$. There are in total $i-1$ steps, in which $n$ could be probed and accepted. Before time $\bar{k}_n$, the optimal algorithm stops with reward $n$ if and only if any of the first $\bar{k}_n-1$ values probed is $n$, because if $b$ were probed before time $\bar{k}_n<i$, the algorithm would not stop by definition of $\bar{k}_n$. Thus the probability of stopping with reward $n$ before time $\bar{k}_n$ is  \[1-\left(1-\frac{1}{n^2}\right)^{\bar{k}_n-1}.\] 
                From time $\bar{k}_n$ onward, at each step, the optimal algorithm stops with reward $n$ if and only if all previous values probed, starting from time $\bar{k}_n$, are $0$. In fact as of time $\bar{k}_n$, if $b$ were to be probed before $n$, the algorithm would stop with reward $b$ by definition of $\bar{k}_n$. Thus the probability of stopping with reward $n$ at each step from time $\bar{k}_n$ onward (there are $i-\bar{k}_n$ such steps), is the probability of not probing $n$ in any of the first $\bar{k}_n-1$ steps (which ensures that the algorithm reaches the $\bar{k}_n$th step), multiplied by the probability of always probing $0$, from time $\bar{k}_n$ up to the first time $n$ is probed:
				\[\frac{1}{n^2}\left(1-\frac{1}{n^2}\right)^{\bar{k}_n-1}\:\sum_{j=0}^{i-\bar{k}_n-1}\left(1-\frac{p}{n}-\frac{1}{n^2}\right)^j.\]
                There are in total $i-1$ steps, in which $b$ could be probed, but only in $i-\bar{k}_n$ of these, from time $\bar{k}_n$ to time $i-1$, $b$ would be accepted. From the previous comments on the acceptance of $b$, at each of these steps, the probability that the optimal algorithm stops with reward $b$ is the probability of not probing $n$ in any of the first $\bar{k}_n-1$ steps (which ensures that the algorithm reaches the $\bar{k}_n$th step), multiplied by the probability that all remaining probed values are $0$, from time $\bar{k}_n$ up to the first time $b$ is probed. Thus the probability of stopping with $b$ is 
                \[
                    \frac{p}{n}\left(1-\frac{1}{n^2}\right)^{\bar{k}_n-1}\:\sum_{j=0}^{i-\bar{k}_n-1}\left(1-\frac{p}{n}-\frac{1}{n^2}\right)^j \,.
                \]
                From the above it follows also that the probability of stopping with reward $a$ is the probability of not probing $n$ in any of the first $\bar{k}_n-1$ steps, multiplied by the probability of always probing $0$ in all the remaining $i-\bar{k}_n$ steps, that is 
                \[
                    \left(1-\frac{1}{n^2}\right)^{\bar{k}_n-1}\left(1-\frac{p}{n}-\frac{1}{n^2}\right)^{i-\bar{k}_n} \,.
                \] 
                In conclusion for every $i> j_n$,
				\begin{equation}\label{ii:i>j}
					(V_{\pi_T}|\Omega_i)\sim\begin{cases}n,&\text{w.p. }1-\left(1-\frac{1}{n^2}\right)^{\bar{k}_n-1}+\frac{1}{n^2}\left(1-\frac{1}{n^2}\right)^{\bar{k}_n-1}\sum_{j=0}^{i-\bar{k}_n-1}\left(1-\frac{p}{n}-\frac{1}{n^2}\right)^j\\b,&\text{w.p. }\frac{p}{n}\left(1-\frac{1}{n^2}\right)^{\bar{k}_n-1}\sum_{j=0}^{i-\bar{k}_n-1}\left(1-\frac{p}{n}-\frac{1}{n^2}\right)^j \\a,&\text{w.p. }\left(1-\frac{1}{n^2}\right)^{\bar{k}_n-1}\left(1-\frac{p}{n}-\frac{1}{n^2}\right)^{i-\bar{k}_n},\end{cases}
				\end{equation}
				and therefore, through the usual estimation methods, by \Cref{accept}(c)
				\begin{align}\label{ii:Omega_i>j}
					\E_iV_{\pi_T}&=n\left[\frac{\bar{k}_n}{n^2}+\mathcal{O}\left(\frac{1}{n^2}\right)\right]+\left[1-\frac{\bar{k}_n}{n^2}+\mathcal{O}\left(\frac{1}{n^2}\right)\right]\left[\left(\frac{1}{p}+b\right)\left(1-e^{p\left(\frac{\bar{k}_n}{n}-\frac{i}{n}\right)}\right)\right]\notag\\&+a\left(1-\frac{\bar{k}_n}{n^2}+\mathcal{O}\left(\frac{1}{n^2}\right)\right)\left(e^{p\left(\frac{\bar{k}_n}{n}-\frac{i}{n}\right)}+\mathcal{O}\left(\frac{1}{n}\right)\right)+\mathcal{O}\left(\frac{1}{n}\right)\notag\\&=\frac{\bar{k}_n}{n}+\left(\frac{1}{p}+b\right)\left(1-e^{p\left(\frac{\bar{k}_n}{n}-\frac{i}{n}\right)}\right)+ae^{p\left(\frac{\bar{k}_n}{n}-\frac{i}{n}\right)}+\mathcal{O}\left(\frac{1}{n}\right).
				\end{align} 
			\end{enumerate}
		
			Having exhausted all cases, by the law of total expectation we can compute, by adopting empty sum convention,
			\begin{align}\label{ii:lawtot_ka<kc<j}
				\E V_{\pi_T}&=\frac{1}{n+1}\sum_{i=1}^{k_n-1}\E_i V_{\pi_T}+\frac{1}{n+1}\sum_{i=k_n}^{\bar{k}_n-1}\E_i V_{\pi_T}+\frac{1}{n+1}\sum_{i=\bar{k}_n}^{j_n-1}\E_i V_{\pi_T}+\frac{1}{n+1}\sum_{i=j_n}^{n+1}\E_i V_{\pi_T}\notag\\&\defeq S_{1,k_n-1}+S_{k_n,\bar{k}_n-1}+S_{\bar{k}_n,j_n-1}+S_{j_n,n+1}.
			\end{align}
            \begin{enumerate}[i)]
		\item We start by calculating $S_{1,k_n-1}$. By \Cref{ii:Omega_i<ka}, denoting as $\mu_n \defeq \sfrac{k_n}{n}$ (note that by \Cref{accept} (b) this quantity is subunitary and bounded away from zero and one as $n\longrightarrow\infty$), we have that
			\begin{align*}
				&\frac{1}{n+1}\sum_{i=1}^{k_n-1}\left[\frac{k_n}{n}+\left(\frac{1}{p}+b\right)\left(1-e^{p\left(\frac{k_n}{n}-1\right)}\right)+\mathcal{O}\left(\frac{1}{n}\right)\right]\\&=\frac{k_n-1}{n+1}\frac{k_n}{n}+\frac{k_n-1}{n+1}\left(\frac{1}{p}+b\right)\left(1-e^{p\left(\frac{k_n}{n}-1\right)}\right)+\mathcal{O}\left(\frac{1}{n}\right),
			\end{align*}
			so it follows that
			\begin{equation}\label{ii:S1ka}
				S_{1,k_n-1}=\mu_n^2+\mu_n\left(\frac{1}{p}+b\right)\left(1-e^{p\left(\mu_n-1\right)}\right)+\mathcal{O}\left(\frac{1}{n}\right).
			\end{equation}
            \item Next we calculate $S_{k_n,\bar{k}_n-1}$. By \Cref{ii:Omega_ka<i<kc}, denoting $\nu_n \defeq \sfrac{\bar{k}_n}{n}$ (note that by \Cref{accept} (a, b, d, e) this quantity is subunitary and bounded away from zero and one as $n\longrightarrow\infty$), we have that
			\begin{align*}
				&\frac{1}{n+1}\sum_{i=k_n}^{\bar{k}_n-1}\left[\frac{i}{n}+\left(\frac{1}{p}+b\right)\left(1-e^{p\left(\frac{i}{n}-1\right)}\right)+\mathcal{O}\left(\frac{1}{n}\right)\right]=\frac{\bar{k}_n(\bar{k}_n-1)-k_n(k_n-1)}{2n(n+1)}\\&+\left(\frac{1}{p}+b\right)\frac{\bar{k}_n-k_n}{n+1}-\frac{1}{n+1}\left(\frac{1}{p}+b\right)\sum_{i=k_n}^{\bar{k}_n-1}e^{p\left(\frac{i}{n}-1\right)}+\mathcal{O}\left(\frac{1}{n}\right)\\&=\frac{\nu_n^2}{2}-\frac{\mu_n^2}{2}+\left(\frac{1}{p}+b\right)(\nu_n-\mu_n)-\frac{e^{-p}}{n+1}\left(\frac{1}{p}+b\right)\frac{e^{p\mu_n}-e^{p\nu_n}}{1-e^\frac{p}{n}}+\mathcal{O}\left(\frac{1}{n}\right).
			\end{align*}
			Therefore, since 
			\[\frac{1}{n+1}\frac{e^{p\mu_n}-e^{p\nu_n}}{1-e^\frac{p}{n}}=\frac{1}{n+1}\frac{e^{p\mu_n}-e^{p\nu_n}}{-\frac{p}{n}+\mathcal{O}\left(\frac{1}{n^2}\right)}=\frac{e^{p\mu_n}-e^{p\nu_n}}{-p+\mathcal{O}\left(\frac{1}{n}\right)}=\frac{e^{p\nu_n}-e^{p\mu_n}}{p}+\mathcal{O}\left(\frac{1}{n}\right),\]
			it follows that
			\begin{equation}\label{ii:Skakc}
				S_{k_n,\bar{k}_n-1}=\frac{\nu_n^2}{2}-\frac{\mu_n^2}{2}+\left(\frac{1}{p}+b\right)(\nu_n-\mu_n)-\frac{e^{-p}}{p}\left(\frac{1}{p}+b\right)(e^{p\nu_n}-e^{p\mu_n})+\mathcal{O}\left(\frac{1}{n}\right).
			\end{equation}
            \item Next we calculate $S_{\bar{k}_n,j_n-1}$. Denoting as $\lambda_n \defeq \sfrac{j_n}{n}$ (note that by \Cref{accept} (a) this quantity is subunitary and bounded away from zero and one as $n\longrightarrow\infty$), by \Cref{ii:Omega_kc<i<j} we have that
			\begin{align*}
				&\frac{1}{n+1}\sum_{i=\bar{k}_n}^{j_n-1}\left[\frac{\bar{k}_n}{n}+\left(\frac{1}{p}+b\right)\left(1-e^{p\left(\frac{\bar{k}_n}{n}-1\right)}\right)+\mathcal{O}\left(\frac{1}{n}\right)\right]\\&=\frac{(j_n-\bar{k}_n)\bar{k}_n}{(n+1)n}+\frac{j_n-\bar{k}_n}{n+1}\left(\frac{1}{p}+b\right)\left(1-e^{p\left(\frac{\bar{k}_n}{n}-1\right)}\right)+\mathcal{O}\left(\frac{1}{n}\right)\\&=(\lambda_n-\nu_n)\nu_n+(\lambda_n-\nu_n)\left(\frac{1}{p}+b\right)\left(1-e^{p\left(\nu_n-1\right)}\right)+\mathcal{O}\left(\frac{1}{n}\right),
			\end{align*}
			so it follows that
			\begin{equation}\label{ii:Skcj}
				S_{\bar{k}_n,j_n-1}=-\nu_n^2+\lambda_n\nu_n+(\lambda_n-\nu_n)\left(\frac{1}{p}+b\right)\left(1-e^{p\left(\nu_n-1\right)}\right)+ \mathcal{O}\left(\frac{1}{n}\right).
			\end{equation}
		\item Finally we compute $S_{j_n,n+1}$. By \Cref{ii:Omega_i>j} we have that
			\[\frac{1}{n+1}\sum_{i=j_n}^{n+1}\left[\frac{\bar{k}_n}{n}+\left(\frac{1}{p}+b\right)\left(1-e^{p\left(\frac{\bar{k}_n}{n}-\frac{i}{n}\right)}\right)+ae^{p\left(\frac{\bar{k}_n}{n}-\frac{i}{n}\right)}+\mathcal{O}\left(\frac{1}{n}\right)\right]\] can be expanded as \begin{align*}&\frac{\bar{k}_n(n+2-j_n)}{n(n+1)}+\frac{1}{n+1}\left(\frac{1}{p}+b\right)\sum_{i=j_n}^{n+1}\left(1-e^{p\left(\frac{\bar{k}_n}{n}-\frac{i}{n}\right)}\right)+\\&\frac{a}{n+1}\sum_{i=j_n}^{n+1}e^{p\left(\frac{\bar{k}_n}{n}-\frac{i}{n}\right)}+\mathcal{O}\left(\frac{1}{n}\right)=\nu_n-\lambda_n\nu_n+\left(\frac{1}{p}+b\right)\frac{n+2-j_n}{n+1}-\\&\frac{1}{n+1}\left(\frac{1}{p}+b-a\right)\sum_{j=j_n-\bar{k}_n}^{n+1-\bar{k}_n}e^{-\frac{pj}{n}} +\mathcal{O}\left(\frac{1}{n}\right)=\nu_n-\lambda_n\nu_n+\left(\frac{1}{p}+b\right)(1-\lambda_n)-\\&\frac{1}{n+1}\left(\frac{1}{p}+b-a\right)\left(\frac{e^{-p\frac{j_n-\bar{k}_n}{n}}-e^{-p\frac{n+2-\bar{k}_n}{n}}}{1-e^{-\frac{p}{n}}}\right)+\mathcal{O}\left(\frac{1}{n}\right)=\nu_n-\lambda_n\nu_n+\\&\left(\frac{1}{p}+b\right)(1-\lambda_n)-\frac{1}{n+1}\left(\frac{1}{p}+b-a\right)\frac{e^{-p(\lambda_n-\nu_n)}-e^{-p(1-\nu_n)}+\mathcal{O}\left(\frac{1}{n}\right)}{1-e^{-\frac{p}{n}}}+\mathcal{O}\left(\frac{1}{n}\right).
			\end{align*}
			Therefore, since 
			\begin{align*}
				&\frac{1}{n+1}\frac{e^{-p(\lambda_n-\nu_n)}-e^{-p(1-\nu_n)}+\mathcal{O}\left(\frac{1}{n}\right)}{1-e^{-\frac{p}{n}}}=\frac{1}{n+1}\frac{e^{-p(\lambda_n-\nu_n)}-e^{-p(1-\nu_n)}+\mathcal{O}\left(\frac{1}{n}\right)}{\frac{p}{n}+\mathcal{O}\left(\frac{1}{n^2}\right)}\\&=\frac{1}{p}\frac{e^{-p(\lambda_n-\nu_n)}-e^{-p(1-\nu_n)}}{1+\mathcal{O}\left(\frac{1}{n}\right)}+\mathcal{O}\left(\frac{1}{n}\right)=\frac{1}{p}(e^{-p(\lambda_n-\nu_n)}-e^{-p(1-\nu_n)})+\mathcal{O}\left(\frac{1}{n}\right),
			\end{align*}
			it follows that
			\begin{equation}\label{ii:Sjn}
				S_{j_n,n+1}=-\lambda_n\nu_n+\nu_n+\left(\frac{1}{p}+b\right)(1-\lambda_n)-\frac{1}{p}\left(\frac{1}{p}+b-a\right)(e^{-p(\lambda_n-\nu_n)}-e^{-p(1-\nu_n)})+\mathcal{O}\left(\frac{1}{n}\right).
			\end{equation}
        \end{enumerate}
			Plugging \Cref{ii:S1ka,ii:Skakc,ii:Skcj,ii:Sjn} into \Cref{ii:lawtot_ka<kc<j} yields, after a few cancellations and collecting of common factors, that
			\begin{align}\label{ii:expectation}
				\E V_{\pi_T}&=\frac{\mu_n^2}{2}-\frac{\nu_n^2}{2}+\nu_n+\frac{1}{p}+b+\left(\frac{1}{p}-\mu_n\right)\left(\frac{1}{p}+b\right)e^{p(\mu_n-1)}\notag\\&+\left[\left(\frac{1}{p}+b\right)(\nu_n-\lambda_n)-\frac{a}{p}\right]e^{p(\nu_n-1)}-\frac{1}{p}\left(\frac{1}{p}+b-a\right)e^{p(\nu_n-\lambda_n)}+\mathcal{O}\left(\frac{1}{n}\right)\notag\\&\defeq q_{a,b,p}(\lambda_n,\mu_n,\nu_n)+\mathcal{O}\left(\frac{1}{n}\right).
			\end{align}
	\end{proof}
	
	\section{Numerical approximations}\label{code}

    \paragraph{Root finding.}
	In this section the code for the computation of $M(a,b,p)\le 0.7235$ is shared. The  list of the first eleven correct decimals is $0.72348603329$, and the value $0.7235$ provided for the hardness is a rounded-up approximation, following from the parameters set in the bisection method. In the Python method used we set $\text{xtol}=10^{-13}$ and $\text{rtol}=10^{-14}$, and the bisection method is implemented such that when it stops, if $\hat{\nu}$ is the root found and $\nu_*$ is the true root of $\tilde{q}'(\nu)$, it is guaranteed that \[|\hat{\nu}-\nu_*|\leq\text{xtol}+|\hat{\nu}|\cdot\text{rtol}.\]
	Given that in our case we found $\hat{\nu}= 0.211231196923\ndots<0.3$, this ensures
	\[|\hat{\nu}-\nu_*|<10^{-13}+0.3\cdot10^{-14}<0.5\cdot10^{-13}.\]
	Moreover the function $q(\nu)$ has derivative close to zero at $\hat{\nu}$, since $\nu_*$ is a stationary point. In particular, for any $\zeta$ between $\hat{\nu}$ and $\nu_*$, $|q'(\zeta)|<1$ (this can be easily verified by direct computation), and therefore by the Lagrange remainder formula for Taylor approximation, we have \[|q(\hat{\nu})-\tilde{q}(\nu_*)|\leq|q'(\zeta)||\hat{\nu}-\nu_*|<|\hat{\nu}-\nu_*|<0.5\cdot10^{-13}.\]
	This largely ensures the correctness of the first eleven digits provided in the decimal approximation of $M(a,b,p)$, since the few other direct computations involved are performed at machine precision (which is slightly less than $2.22\cdot10^{-16}$ in Python).

    \paragraph{Numerical computations.}
    The Python codes used for numerical computations are available at \url{https://doi.org/10.5281/zenodo.10649575}.
    \begin{itemize} 
    \item The code bisection.py approximates the value of the upper bound $M(a,b,p)$ on the gambler-to-prophet ratio of the optimal algorithm for Instance \ref{instance}. 
    \item The code DP.py simulates the dynamic program used to produce \Cref{prophet} and an approximation of the gambler-to-prophet ratio of the optimal algorithm for Instance \ref{instance} with $n=10^6$, so as to estimate the sharpness of the value obtained for $M(a,b,p)$.
    \end{itemize}

	\section{Optional supplements to Section \ref{Preliminaries}}\label{suppreview}
	\subsection{Supplements to Section \ref{BIRO}}\label{suppbackwardcomp}
	In this section we provide the details of the rigorous measure-theoretic construction for (finite) random order processes, so as to extend \cite[Theorem~3.2]{ChoRobSieg71} to them. In order to do this formally, it is sufficient to define a filtration for the process. Let us start by showing how a randomly ordered list can be treated consistently as a process (note that the uniformity of the random order is not even needed, even though for simplicity of exposition we will assume it).
	
	\paragraph{Formal construction of the random order process.} Recall that $S_{n+1}$ denotes the set of permutations of $[n+1]$ and $\mathbb{R}_+^{n+1}\defeq\{x\in\mathbb{R}^{n+1}:\;x_i\geq 0, \; \forall\, 1\leq i\leq n+1\}$.
	We denote by $2^{S_{n+1}}$ the $\sigma$-algebra of subsets of $S_{n+1}$ and by $\mathcal{B}(\mathbb{R}_+^{n+1})$ the Borel $\sigma$-algebra on $\mathbb{R}_+^{n+1}$. Let $(\Omega,\mathcal{F}, \Prob )$ be the probability space supporting jointly $\pi$ and $X=(V_1,\ldots,V_{n+1})$, with $(\pi,\: X)$ valued in the product measure space \[(S_{n+1}\times\mathbb{R}_+^{n+1},\: 2^{S_{n+1}}\otimes \mathcal{B}(\mathbb{R}_+^{n+1}),\: \mu\otimes\nu),\] where: 
	\begin{itemize}[noitemsep]
		\item for every $\pi\in S_{n+1}$, $\mu(\pi) \equiv \frac{1}{(n+1)!}$ is the uniform probability law on the permutations; 
		\item $\nu$ is the law of $X$, that is the joint law of $V_1,\ldots,V_{n+1}$; 
		\item $\otimes$ denotes the standard product measure and $\sigma$-algebra.
	\end{itemize}
	Note that a random permutation can be seen as a stochastic process $\pi=(\pi_1,\ldots,\pi_n)$. This is a very natural thing to do. Think of the (uniform) law on $S_{n+1}$ as derived from the conditional laws 
	\begin{align*}
		\pi_1&\sim \Unif([n+1]).\\(\pi_2|\pi_1)&\sim \Unif([n+1]\setminus\{\pi_1\}),\\&\vdots\\(\pi_n|\pi_{n-1},\ldots,\pi_1)&\sim \Unif([n+1]\setminus\{\pi_1,\ldots,\pi_{n-1}\}),\\ (\pi_{n+1}|\pi_n,\ldots,\pi_1)&\sim\delta_{[n+1]\setminus\{\pi_1,\ldots,\pi_{n}\}},
	\end{align*} 
	where $\Unif(\cdot)$ denotes the discrete uniform distribution and $\delta_x$ is the Dirac distribution centred at $x$.\footnote{It is not necessary to have the uniform law in the general case treated here. We just adopt it for simplicity, since it will be the law used for RO. Any random ordering, depending on the model, can be similarly described and implemented within this general framework.} As a finite stochastic process, $\pi$ is the measurable map 
	\begin{align*}
		\pi:\:\Omega\times[n+1]&\longrightarrow [n+1]\\ (\omega,i)&\mapsto \pi_i(\omega)
	\end{align*} 
	and can be composed with the vector $X=(V_1,\ldots,V_{n+1})$, which can also be seen as a finite stochastic process with the joint laws previously given, that is a measurable map 
	\begin{align*}
		X:\:\Omega\times[n+1]&\longrightarrow\mathbb{R}_+\\(\omega,i)&\mapsto V_i(\omega).
	\end{align*} 
	The composition needs to be done by exploiting the extended \emph{graph map} associated with $\pi$, which is measurable and will be denoted as 
	\begin{align*}
		\tilde{\pi}:\:\Omega\times[n+1]&\longrightarrow \Omega\times[n+1]\\ (\omega,i)&\mapsto (\omega,\pi_i(\omega)).
	\end{align*}
	Due to the measurability of all maps involved (which clearly follows from the time index set being discrete), the result of composing $X$ and $\tilde{\pi}$ is the randomly permuted finite stochastic process $X^\pi\defeq X\circ\tilde{\pi}$ and corresponds, by going back to the usual random vector point of view, to the randomly permuted random vector of components $X_i^\pi=V_{\pi_i}$ for all $i\in [n+1]$.
	
	\paragraph{Formal construction of the random order filtration.} For every $B\in\mathcal{B}(\mathbb{R}_+)$, $V_{\pi_k}\in B$ if and only if for every $k\in [n+1]$, $\pi_k=i$ and $V_i\in B$. Therefore \[\{V_{\pi_k}\in B\}=\bigcap_{i=1}^{n+1}(\{\pi_k=i\}\cap\{V_i\in B\}).\] Equivalently \[V_{\pi_k}^{-1}(B)=\bigcap_{i=1}^{n+1}(\pi_k^{-1}(i)\cap V_i^{-1}(B)),\] and we can thus define the $\sigma$-algebra generated by $V_{\pi_k}$ as \[\sigma(V_{\pi_k})\defeq \sigma(\{F\in\mathcal{F}: F=V_{\pi_k}^{-1}(B),\;B\in\mathcal{B}(\mathbb{R}_+)\}).\] A filtration that accounts for the random order arrival of a general sequence of random variables $V_1\ldots,V_{n+1}$ with undisclosed RO can be constructed by considering $\mathcal{F}_k^{n+1}$ (which will be denoted simply as $\mathcal{F}_k$, by omitting the finite horizon) defined as follows: $\mathcal{F}_0\defeq\{\emptyset,\Omega\}$, and for all $k\in [n+1]$, $\mathcal{F}_k\defeq\sigma(V_{\pi_1},\ldots,V_{\pi_k})$. In disclosed RO one would have to interleave with $\pi$, that is $\mathcal{F}_k\defeq\sigma(V_{\pi_1},\pi_1,\ldots,V_{\pi_k},\pi_k)$. For simplicity we adopt the former filtration and work with the undisclosed model.
	
	\paragraph{Backward induction for random order processes.} Given a stopping time $\tau$, that is a random variable $\tau:\Omega\longrightarrow [n+1]$ such that $\{\tau=k\}\in\mathcal{F}_k$, we can also define the stopped version of the process $X^\pi$. Consider first the stopped process $\pi^\tau\defeq\{\pi_{i\wedge \tau},\:i\in[n+1]\}$ and then define $X^{\pi^\tau}\defeq X\circ\tilde{\pi}^\tau$, where $\tilde{\pi}^\tau:(\omega,i)\mapsto(\omega,\pi_{i\wedge \tau(\omega)}(\omega))$. Thus $X^{\pi^\tau}_i= V_{\pi_{i\wedge \tau}}$, so that the stopped values of $X^{\pi^\tau}$ is $V_{\pi_{\tau}}$. We will abuse the notation and simply refer to $X^\pi$ as $X$. Thus $X_i\defeq X_i^\pi=V_{\pi_i}$ for all $i\in [n+1]$ and the conditional expectation on the filtration $\E_{\mathcal{F}_k}(\cdot)$ can be equivalently denoted as $\E(\,\cdot\,|X_1,\ldots,X_k)$. To construct the optimal stopping rule we define $\gamma_{n+1} \defeq V_{\pi_{n+1}}$ and, for all $l \in [n]$,
	\[
	\gamma_l \defeq \max \{ V_{\pi_l}, \E_{\mathcal{F}_l}\gamma_{l+1}\} \,.
	\] 
	We also define, for each $k \in [n+1]$,
	\[
	s_k \defeq \inf \{ l\geq k : \; V_{\pi_l} = \gamma_l \}\,.
	\]
	Then we automatically have the following result, as an extension of \cite[Theorem~3.2]{ChoRobSieg71}. The proof remains the same, upon a formal change of their filtration with ours.
	\begin{theorem}\label{backwardcomp}
		For a given instance $V_1, V_2, \ldots, V_{n + 1}$ denote the random arrival order process $X\defeq(V_{\pi_1},\ldots,V_{\pi_{n+1}})$, and consider the backward induction values $ \gamma_{n+1}, \gamma_n, \ldots, \gamma_1$ and stopping times $s_{n+1}, s_n, \ldots, s_1$ as previously defined. Then for all $k\in[n+1]$, $s_k\in C_k$ and
		\[\E_{\mathcal{F}_k}V_{\pi_{s_k}}=\gamma_k\geq \E_{\mathcal{F}_k} V_{\pi_\tau}\] for all stopping rules $\tau\in C_k$. Taking expectations yields  \[\E V_{\pi_{s_k}}=\E\gamma_k\geq \E V_{\pi_\tau}\] for all stopping rules $\tau\in C_k$. Thus $\E\gamma_k=\sup_{\tau\in C_k}\E V_{\pi_\tau}$.
	\end{theorem}

	\subsection{Supplements to Section \ref{application}}\label{suppbackward}
	In this section we describe a stopping rule for the discrete setting more explicitly, slightly adapting \cite[\S 7.1]{Ferg67} 
	to random order processes, in order to show how to compute the reward of the optimal stopping rule directly, and provide evidence for the claim that it is not a feasible option in our setting. We also comment on the fact that there is no loss of generality in working only with nonrandomised stopping rules, such as those yielded by backward induction.
	
	Recall from \Cref{suppbackwardcomp} that the distribution of the random variable $X_i$, denoting the $i$-th component of the random order process $X=(V_{\pi_1},\ldots,V_{\pi_{n+1}})$, is known jointly with the distribution of $\pi_i$, since by our model, regardless of whether $\pi_i$ is disclosed or not, \[(X_i|\pi_i=\sigma_i)\sim V_{\sigma_i},\] denoting the values taken by $X_i$ as $x_i$ and those taken by $\pi_i$ as $\sigma_i$. 
	
	A (possibly randomised) stopping rule for a sequence of deterministic reward functions $\{y_1(x_1),\ldots,y_{n+1}(x_1,\ldots,x_{n+1})\}$ can be described as a sequence $\psi=(\psi_1(x_1),\ldots\psi_{n+1}(x_1,\ldots,x_{n+1}))$ where $\psi_i=\psi_i(x_1,\ldots,x_i)$ is the probability of stopping at step $i$ given that $i$ observations have been taken, namely $X_1=x_1,\ldots,X_i=x_i$ (for non-randomised stopping rules $\psi_i\in\{0,1\}$ for all $i\in[n+1]$, which is the case of the optimal stopping rule constructed in \Cref{backwarddiscrete}). The stopping rule $\psi$ and the observations of $X$ up to time $i$ only, determine the stopping time $0\leq \tau\leq n+1$. More formally the conditional probability mass function of $\tau$ given $X=x$ can be denoted as $\psi=\{\psi_1,\ldots,\psi_{n+1}\}$ where 
	\[\psi_i(x_1,\ldots,x_i)=\Prob(\tau=i|X_1=x_1,\ldots, X_i=x_i)=\Prob(\tau=i|X=x).\]
	Note that the last equality encodes conditional independence of $\{\tau=i\}$, given the past up to the present, from future observations $X_{i+1},\ldots,X_{n+1}$, ensuring that it is a stopping time (the randomised case can be reduced to a stopping time by enlarging the probability space, as per the concluding remark of this section). By the assumptions of the model, we also have that
	\begin{align*}
		\psi_i(x_1,\ldots,x_i)&=\Prob(\tau=i|\pi_1=\sigma_1,X_1=x_1,\ldots, \pi_i=\sigma_i,X_i=x_i)\\&=\Prob(\tau=i|\pi=\sigma, X=x)=\Prob(\tau=i|X_1=x_1,\ldots,X_i=x_i),
	\end{align*} 
	since the decision to stop is reached while unaware of the arrival order of the distributions in the past, that is, more precisely, the stopping rule, conditionally on the observed values, is independent of the permutation (the last equality need not hold in general in disclosed RO). The conditional probability mass function $\psi$ determines the law of the stopping time $\tau$. Within this framework, optimal stopping is choosing a stopping rule $\psi$ determining a stopping time $T$ that yields maximal expectation of the stopped reward sequence $\E y_T\defeq \E V_{\pi_T}$. Since in our case $y_i(x_1,\ldots,x_i)\defeq x_i$, this expectation is calculated by exploiting $y_T=X_T=\sum_{i=1}^{n+1}X_i\mathbbm{1}_{\{T=i\}}$, yielding
	\begin{align*}
		\E X_T&=\E_{\pi,X}\E(X_T|\pi,X)=\E_{\pi,X}\sum_{i=1}^{n+1}X_i\Prob(T=i|\;\pi,X)=\E_{\pi,X}\sum_{i=1}^{n+1}X_i\psi_i(X_1,\ldots,X_i)\\&\defeq\E\sum_{i=1}^{n+1}V_{\pi_i}\psi_i(V_{\pi_1},\ldots,V_{\pi_i}),
	\end{align*}
	where the abuse of notation $\E_{\pi,X}(\cdot)$ does not stand for conditioning, but for the joint laws with respect to which integration is carried, upon change of variables. Such notation is no longer necessary in the last step, where all sources of randomness are explicitly stated in the integrands' notation. To clarify, all the measures being discrete, turns all integrations into summations:
	\begin{align*}	\E_{\pi,X}&\sum_{i=1}^{n+1}X_i\Prob(T=i|\;\pi,X)\\&=\sum_{\sigma\in S_{n+1}}\sum_{x\in \mathcal{S}_n}\sum_{i=1}^{n+1}x_i\Prob(T=i|\;\pi=\sigma,\;X=x)\Prob(\pi=\sigma,X=x)\\&=\sum_{\sigma\in S_{n+1}}\sum_{x\in \mathcal{S}_n}\sum_{i=1}^{n+1}x_i\Prob(T=i|X=x)\Prob(X=x|\pi=\sigma)\Prob(\pi=\sigma)\\&=\sum_{\sigma\in S_{n+1}}\frac{1}{(n+1)!}\sum_{x\in \mathcal{S}_n}\Prob(V_{\sigma_1}=x_1)\ldots\Prob(V_{\sigma_{n+1}}=x_{n+1})\sum_{i=1}^{n+1}x_i\psi_i(x_1,\ldots,x_i).
	\end{align*} 
	We will not use the abuse of notation aforementioned outside this section.
	
	We conclude this section with a comment on the optimal stopping rule characterised in \Cref{backwarddiscrete}: backward induction yields a non-randomised optimal stopping rule $T$, but this comes with no loss of generality, because for any randomised stopping rule there is a non-randomised equivalent one (in expectation), and it is therefore sufficient to work with an optimal non-randomised stopping rule. To see the equivalence in our case, start with a possibly randomised stopping rule $\psi$. An equivalent non-randomised stopping rule is then yielded by enlarging the product space with $[0,1]^{n+1}$ endowing this factor with the product uniform measure on the unit interval, and considering the interleave process $(\pi , X, U)\defeq\{\pi_1,X_1,U_1,\ldots,\pi_{n+1},X_{n+1},U_{n+1}\}$ with $\{U_i\}$ being iid copies of $U\sim \Unif[0,1]$. With respect to the filtration incorporating the uniform random variables interleaved, that is $\mathcal{F}_k=\sigma(V_{\pi_1},U_1,\ldots,V_{\pi_k},U_k)$, the randomized stopping rule yields a stopping time $\tau$. Define the nonrandomised stopping rule $\tilde{\psi}$ as \[\tilde{\psi}_i(X_1,U_1,\ldots,X_i,U_{i})\defeq\mathbbm{1}_{\{U_i<\psi_i(X_1,\ldots,X_i)\}},\] maintaining the same reward sequence. By exploiting the independence of the uniform random variables introduced, the two stopping rules are equivalent, since if we denote by $\tilde{\tau}$ the stopping time yielded by $\tilde{\psi}$, then by the independence of $U_i$ it follows that
	\begin{align*}
		\E V_{\pi_{\tilde{\tau}}}&=\E_{\pi,X,U}\sum_{i=1}^{n+1}X_i\tilde{\psi}_i(X_1,U_1\ldots,X_i,U_i)=\E_{\pi,X}\sum_{i=1}^{n+1}X_i\E_U\mathbbm{1}_{\{U_k<\psi_i(X_1,\ldots,X_i)\}}\\&=\E_{\pi,X}\sum_{i=1}^{n+1}X_i\psi_i(X_1,\ldots,X_i)=	\E V_{\pi_\tau}.
	\end{align*}
 
\end{document}